\documentclass[letterpaper, 10pt, journal]{IEEEtran}
\usepackage[caption=false,font=normalsize,labelfont=sf,textfont=sf]{subfig}
\usepackage{graphicx} 
\usepackage{blindtext}
\usepackage{amssymb}
\usepackage{amsmath}
\usepackage{bm}
\usepackage{tikz}
\usetikzlibrary{calc}
\usetikzlibrary{arrows.meta, positioning, fit, backgrounds}
\usepackage{pgfplots}
\pgfplotsset{compat=newest}
\usepackage{stfloats}  
\usepackage{algorithm}
\usepackage[noend]{algcompatible} 
\usepackage{algpseudocode}        
\usepackage{glossaries}
\usepackage{float}
\usepackage{xcolor}
\usepackage{siunitx}
\usepackage{cite}
\usepackage{enumitem}
\usetikzlibrary{shapes}
\usetikzlibrary{decorations.pathreplacing}
\usepackage{bbm}
\usetikzlibrary{fit,backgrounds,shapes.misc}
\newcounter{MYtempeqncnt}

\usepackage{shellesc}

\pgfplotsset{compat=1.18}
\newtheorem{theorem}{Theorem}
\newtheorem{proposition}{Proposition}

\input{antStyles/satellite.sty}
\input{antStyles/person.sty}
\input{antStyles/phone.sty}
\input{antStyles/bs.sty}
\input{antStyles/car.sty}

\newcommand{\CdfJ}{3}
\newcommand{\CdfK}{25}
\newcommand{\FixedJ}{2}
\newcommand{\FixedK}{40}

\newcommand{\RequireGeneratedFile}[1]{%
	\IfFileExists{#1}{}{%
		\PackageError{dynamic-fdd-plots}{Missing generated file '#1'}{Compile with -shell-escape or run generate_plot_tables.py manually before compiling.}%
	}%
}

\RequireGeneratedFile{generated/cdf_spin0.csv}
\RequireGeneratedFile{generated/cdf_spin1.csv}
\RequireGeneratedFile{generated/cdf_optimized.csv}
\RequireGeneratedFile{generated/avg_vs_k.csv}
\RequireGeneratedFile{generated/avg_vs_j.csv}
\RequireGeneratedFile{generated/improvement_vs_k.csv}
\RequireGeneratedFile{generated/improvement_vs_j.csv}
\RequireGeneratedFile{generated/improvement_p10_vs_k.csv}
\RequireGeneratedFile{generated/improvement_p10_vs_j.csv}

\pgfplotsset{
	mygrid/.style={
		grid=both,
		major grid style={draw=gray!35},
		minor grid style={draw=gray!20},
		tick align=outside,
		tick pos=left,
		line width=0.95pt
	},
	plotspinzero/.style={
		color=orange!85!black,
		line width = 1pt,
		mark=triangle*,
		mark options={solid, scale=1.6}
	},
	plotspinone/.style={
		color=mygreen,
		line width = 1pt,
		mark=o,
		mark options={solid, scale=1.5, fill= mygreen}
	},
	plotopt/.style={
		color=mydarkred,
		solid,
		mark=none,
		line width = 1.5pt,
		mark options={solid, scale=0.9}
	}
}

\tikzset{
	problem/.style={
		rectangle,
		draw,
		rounded corners,
		minimum width=1.8cm,
		minimum height=0.8cm,
		align=center
	},
	arrow/.style={-{Latex}, thick},
	doublearrow/.style={<->, thick}
}

\makeglossaries

\newacronym{3gpp}{3GPP}{$\text{3}^{\text{rd}}$ generation partnership project}
\newacronym{5g}{5G}{fifth-generation}
\newacronym{6g}{6G}{sixth-generation}

\newacronym{bs}{BS}{base-station}

\newacronym{cdf}{CDF}{cumulative distribution function}
\newacronym{csi}{CSI}{channel state information}

\newacronym{drl}{DRL}{deep-reinforcement learning}
\newacronym{dl}{DL}{downlink}
\newacronym{dof}{DoF}{degree-of-freedom}

\newacronym{ecef}{ECEF}{earth-centered, earth-fixed}
\newacronym{efc}{EFC}{Earth-fixed cell}
\newacronym{emc}{EMC}{Earth-moving cell}

\newacronym{fdd}{FDD}{frequency-division duplex}
\newacronym{fspl}{FSPL}{free-space path loss}

\newacronym{gnb}{gNB}{next-generation node-B}
\newacronym{gnss}{GNSS}{global navigation satellite system}
\newacronym{gso}{GSO}{geostationary Earth orbit}
\newacronym{iot}{IoT}{internet-of-things}

\newacronym{leo}{LEO}{low Earth orbit}
\newacronym{los}{LOS}{line-of-sight}

\newacronym{mab}{MAB}{multi-arm bandit}
\newacronym{meo}{MEO}{medium Earth orbit }
\newacronym{mimo}{MIMO}{multiple-input multiple-output}
\newacronym{minlp}{MINLP}{mixed integer non-linear programming}
\newacronym{milp}{MILP}{mixed integer linear programming}
\newacronym{mrc}{MRC}{maximum-ratio combiner}
\newacronym{mrt}{MRT}{maximum-ratio transmission}

\newacronym{ned}{NED}{North-East-Downward}
\newacronym{neu}{NEU}{North-East-Upward}
\newacronym{ntn}{NTN}{non-terrestrial network}
\newacronym{nlos}{NLOS}{non-line-of-sight}
\newacronym{ngso}{NGSO}{non-geostationary Earth orbit}

\newacronym{ofdm}{OFDM}{orthogonal-frequency division multiplexing}

\newacronym{ran}{RAN}{radio-access network}

\newacronym{sa}{SA}{simulated annealing}
\newacronym{sinr}{SINR}{signal-to-noise plus interference ratio}
\newacronym{soc}{SOC}{second-order cone}

\newacronym{tdd}{TDD}{time-division duplex}

\newacronym{ue}{UE}{user-equipment}
\newacronym{ul}{UL}{uplink}
\newacronym{upa}{UPA}{uniform-planar array}

\newacronym{wrt}{w.r.t.}{with respect to}

\glsdisablehyper  
\glsenableentrycount

\definecolor{myblue}{RGB}{0,104,180}
\definecolor{myred}{RGB}{244,161,152}
\definecolor{mydarkred}{RGB}{157,34,70}
\definecolor{mygreen}{RGB}{0,136,120}
\definecolor{mypantone}{RGB}{59,41,106}

\def\BibTeX{{\rm B\kern-.05em{\sc i\kern-.025em b}\kern-.08em
		T\kern-.1667em\lower.7ex\hbox{E}\kern-.125emX}}
\usepackage{balance}
\allowdisplaybreaks[3]

\allowdisplaybreaks[3]
		\newcommand\submittedtext{%
	\footnotesize This work has been submitted to the IEEE for possible publication. Copyright may be transferred without notice, after which this version may no longer be accessible.}

\newcommand\submittednotice{%
	\begin{tikzpicture}[remember picture,overlay]
		\node[anchor=south,yshift=10pt] at (current page.south) {\fbox{\parbox{\dimexpr0.65\textwidth-\fboxsep-\fboxrule\relax}{\submittedtext}}};
	\end{tikzpicture}%
}

\title{Dynamic FDD for Spectrum Sharing \\ in Non-Terrestrial Networks}

\author{%
	Sourav Mukherjee,~\IEEEmembership{Graduate Student Member,~IEEE}, 
	Bho Matthiesen,~\IEEEmembership{Senior Member,~IEEE}, 
    Armin Dekorsy,~\IEEEmembership{Senior Member,~IEEE}, 
	Petar Popovski,~\IEEEmembership{Fellow,~IEEE}
	\thanks{Part of the paper is accepted for publication in ICC Workshops 2026~\cite{Mukherjee}.}
		\thanks{
		This work was supported by the Federal Ministry of Research, Technology and Space (BMFTR) of Germany under grant 16KIS2428K (6G-Coverage) and 16KIS2408 (Open6GHub+).
	}%
	\thanks{Sourav Mukherjee and Armin Dekorsy are with the
		Gauss-Olbers Center and the Department of Communications Engineering,
		University of Bremen, 28359 Bremen, Germany (e-mail: mukherjee@ant.uni-bremen.de; dekorsy@ant.uni-bremen.de).}
    \thanks{Bho Matthiesen is with the
		Department of Communications Engineering,
		Paderborn University, 33098 Paderborn, Germany (e-mail: matthiesen@nt.uni-paderborn.de).}
	\thanks{Petar Popovski is with the Department of Electronic Systems, Aalborg
		University, 9100 Aalborg, Denmark, and also with the Department of Communications Engineering, University of Bremen, 28359 Bremen, Germany
		(e-mail: petarp@es.aau.dk).}

}
\begin{document}
		\bstctlcite{IEEEexample:BSTcontrol}

	\maketitle
	\submittednotice
	\begin{abstract}
	Future 6G networks are envisioned to integrate low Earth orbit satellite mega-constellations to enable seamless global connectivity, particularly in underserved and remote areas. However, the deployment of dense mega-constellations introduces interference among satellites operating over shared frequency bands. This represents a rather new setup for studying spectrum sharing, which exacerbates the limited flexibility of conventional \cgls{fdd} systems based on fixed bands for downlink and uplink transmissions. We address this spectrum-sharing problem and propose dynamic re-assignment of \cgls{fdd} bands for improved interference management in dense deployments{, as well as} evaluate the performance gain of this approach. To this end, we formulate a joint optimization problem that incorporates dynamic band assignment, user scheduling, and power allocation in both directions. This non-convex mixed integer problem is solved using a combination of equivalence transforms, alternating optimization, and state-of-the-art industrial-grade mixed integer solvers. Numerical results demonstrate that the proposed approach of dynamic \cgls{fdd} band assignment significantly enhances system performance over conventional \cgls{fdd}, achieving up to 30\% improvement in throughput in dense deployments.
\end{abstract}
\glsresetall
	
\section{Introduction}

\IEEEPARstart{T}{he} digital divide continues to separate regions with reliable broadband connectivity from those with limited or no access, and it remains an important motivation for the development of 6G~\cite{Yaacoub2020, Rajatheva2020, Leyva-Mayorga2020, Na2024, GSMA2025}. Satellite communications are widely regarded as a key technology to enable these visions in future 6G networks. By deploying \cgls{gso}, constellations of \cgls{meo} and \cgls{leo} satellites, it becomes possible to extend broadband coverage to users across the globe~\cite{Wu2024}. Among these, \cgls{leo} satellites have recently received a significant attention from both academia and industry~\cite{Wang2022, Luo2024}. Consequently, organizations such as Iridium, OneWeb, SpaceX, Amazon, and the Europe's IRIS$^2$ have already deployed, or announced plans to deploy large-scale \cgls{leo} satellite constellations aimed at providing global broadband connectivity~\cite{Wang2022}.
	
As \cgls{leo} satellite constellations continue to expand, the availability of orthogonal spectral resources becomes a critical bottleneck~\cite{Lee2023, Choi2024}. Traditionally, satellite communication systems operate in designated frequency bands such as the L, S, Ku, and Ka bands. However, the rapid deployment of large-scale \cgls{leo} constellations significantly increases the demand for spectral resources. As the number of satellites grows, many of them cover overlapping geographical regions. In such scenarios, allocating strictly orthogonal frequency bands to each satellite becomes increasingly limiting. Consequently, multiple satellites are often required to operate within the same frequency bands, which introduces new challenges related to interference management and efficient spectrum utilization.
	
Additionally, an operator often has licenses in multiple frequency bands, e.g.,  in sub-6 GHz and in the upper mid-band between 7–24 GHz. Although these bands are available at the network, individual satellites are typically configured to operate within a single pre-defined band. Consequently, the available spectral diversity cannot be fully exploited at the satellite level. Recent advances in satellite payload architectures have begun to address this limitation by enabling multi-band satellite operation~\cite{Aboagye2024}. With these developments, a satellite can transmit and receive across multiple bands, allowing the system to select the most suitable band depending on network conditions and operational requirements~\cite{Nie2019, YuanTang2024}. However, the potential benefits of multi-band capability are still under-explored. In most existing systems, each frequency band is partitioned to support both \cgls{ul} and \cgls{dl} transmissions, following the principle of \cgls{fdd}, where separate frequency bands are reserved for each transmission direction. While this structure simplifies system design, it reduces the flexibility with which spectral resources can be utilized across multiple bands. Furthermore, 6G is expected to support heterogeneous services which results in asymmetric bandwidth requirements for \cgls{ul} and \cgls{dl}. Notably, artificial intelligence (AI) applications are expected to change the \cgls{ul}/\cgls{dl} requirements, in many applications favoring the \cgls{ul}~\cite{ngmn2026aiSurge}. Under such conditions, a rigid band selection with static partitioning to support \cgls{ul} and \cgls{dl} is inefficient and limiting. 

To address these limitations, we propose a flexible band assignment strategy in which an entire frequency band is dedicated to a specific transmission direction, either \cgls{ul} or \cgls{dl}, selected from the available frequency bands. This assignment can be dynamically adapted according to traffic demand, interference conditions, and network requirements. By decoupling band selection from fixed duplexing partitions, the system gains additional flexibility in utilizing multi-band resources while preserving the fundamental principle of \cgls{fdd}. Since the proposed approach maintains this separation while enabling adaptive directional selection across bands, we refer to our scheme as \textit{dynamic \cgls{fdd}}. This method introduces an additional \cgls{dof} for multi-band satellite systems and provides a structured mechanism to improve spectrum utilization in satellite mega-constellations.

\begin{figure}[t]
	\centering
	\usetikzlibrary{arrows.meta, positioning, fit}
	\input{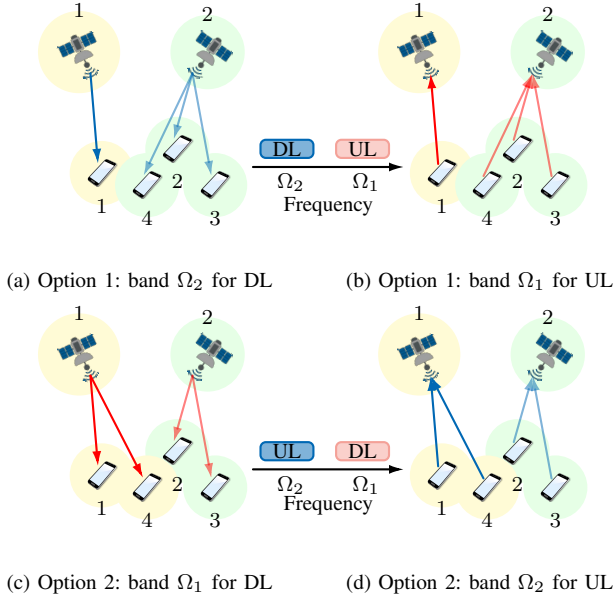}
	
	\caption{Illustration of dynamic \cgls{fdd} in a two-satellite scenario. Two FDD configurations out of four across bands $\Omega_1$ and $\Omega_2$ are depicted, where UL and DL directions are interchanged. This degree-of-freedom modifies how links are scheduled over the available bands and consequently changes system characteristics. This enables more favorable link arrangements and improves spectral utilization compared to conventional FDD, where the band assignment is static.}
	\label{fig:1}	
\end{figure}

\subsection{Dynamic FDD}
To explain the benefit of dynamic \cgls{fdd}, consider a scenario with two satellites as illustrated in Fig.~\ref{fig:1}. It illustrates two duplexing configurations, options 1 and 2, out of four possible configurations across the bands $\Omega_1$ and $\Omega_2$, where \cgls{ul} and \cgls{dl} directions can be interchanged. Option~1 is depicted in Fig.~\ref{fig:1}(a)-(b), where satellite~$1$ serves \cgls{ue}~$1$ and satellite~$2$ serves the remaining users.  With the new \cgls{dof} of band switching, we can create alternative configurations, such as option~2, Fig.~\ref{fig:1}(c)-(d), where satellite~$1$ serves \cglspl{ue}~$1$ and $2$, and satellite~$2$ serves the remaining users. Therefore, with the introduction of \cgls{dof}, the scheduling decisions can be re-configured. The key observation from Fig.~\ref{fig:1} is that the duplexing decision not only determines the operating frequency of a link but also governs how different links interact to each other. In particular, the interference patterns of the system change depending on how \cgls{ul} and \cgls{dl} transmissions are mapped onto the available bands. As seen in the two options, switching the band assignment effectively reshapes the interference pattern across the network. However, the benefit of dynamic \cgls{fdd} is not limited to selecting the band with lower interference. More fundamentally, it introduces an additional \cgls{dof} that allows the systems to create alternative link configurations.  As a result, the system can avoid unfavorable configurations where strong interfering links are forced to coexist in the same band and direction. From a resource allocation perspective, this additional \cgls{dof} enlarges the feasible set of link configurations. Instead of operating under a fixed duplexing structure, the network can adapt the \cgls{ul}-\cgls{dl} split to better match the spatial distribution of \cglspl{ue} and satellites. Consequently, dynamic \cgls{fdd} improves performance through two coupled effects. First, it allows the system to select the more favorable band for each transmission direction. Second, and more importantly, it enables a reorganization of link activity across bands, which reduces harmful interference coupling and leads to more efficient utilization of spectral resources. This additional flexibility is particularly beneficial in dense or heterogeneous deployments, where static duplexing configurations are inherently limiting.

\subsection{Related Work}
    	
 Early works largely focus on the coexistence between \cgls{gso} and \cgls{ngso} systems, consisting of \cgls{meo} and \cgls{leo}. For instance, methods such as range-based beam control, cognitive radio, and traffic-aware power allocation are proposed in~\cite{Pourmoghadas2017} to manage the in-line interference when \cgls{ngso} satellites cross the \cgls{los} between a \cgls{gso} satellite and its ground station. Other techniques, including exclusive-angle strategies \cite{Wang2018} and look-aside~\cite{Braun2019}, provide useful insights for interference management in low-density deployments. On the other hand, for dense satellite deployments, aggregated interference at the receiver becomes a challenge. Studies show that \cgls{gso} systems can suffer severely from \cgls{ngso} transmissions, highlighting the need for stronger interference mitigation strategies \cite{Wang2020, Jalali2024}.
	 
More recent studies shift attention to the critical problem of \cgls{ngso}-\cgls{ngso} coexistence, and in particular \cgls{leo}-\cgls{leo} interference, which becomes increasingly severe in dense constellations~\cite{Al-Hraishawi2023}. To avoid excessive interference between multiple adjacent~\cgls{ngso} satellites, joint optimization-based methods are proposed in~\cite{Al-Hraishawi2023}. Beam allocation approaches based on matching theory, aiming to maximize throughput while mitigating cross-constellation interference between two constellations, are proposed in~\cite{Zhang2024}; joint beamforming and satellite selection methods are developed to improve~\cgls{dl} performance and positioning accuracy simultaneously in~\cite{Li2024}. Long-term optimization frameworks that jointly consider beam direction, power, frequency, and time resources are also proposed to enhance spectrum efficiency in~\cite{Yuan2024}. At the constellation coordination level, satellite selection methods in~\cite{Kim2025a, Kim2025b} allow secondary systems to adaptively serve users while protecting primary systems from harmful interference, making in-band coexistence feasible. Joint transmit power and beam directivity control, enabled by digital beamforming, is also shown to be effective in reducing inter-beam interference~\cite{Choi2005, Xie2023, Takahashi2019}. Similarly, adjusting spot-beam sizes according to user demand provides additional flexibility for interference management~\cite{Lei2024}. Further, joint beamforming and power control strategies are investigated, often coupled with co-frequency exclusion zones to mitigate cross-system interference in \cite{Yin2021, Gu2022}. Additionally, beam-hopping methods, where different frequency bands are allocated to different beams of a multi-beam satellite, help reduce intra-beam interference, as shown in \cite{Jia2025, Tang2021}. To further enhance spectrum sharing efficiency, database-driven spectrum management is suggested for coexistence in dynamic operating environments~\cite{Höyhtyä2017}. These methods highlight the trend toward adaptive and flexible designs that dynamically respond to interference conditions in dense satellite deployments.
	
On the other hand, a significant body of work focuses on the interaction and coexistence between satellites and terrestrial networks. In such hybrid systems, various resource allocation mechanisms are proposed, such as~\cite{Ding2024}. A promising approach lies in spectrum pairing strategies, where normal and reverse pairing between terrestrial and \cgls{leo} systems enable more flexible utilization of \cgls{ul} and \cgls{dl} bands~\cite{Lee2023, Fu2023, Lee2024, Zhang2019}. In the normal pairing mode, both systems use the same band for \cgls{dl} and another shared band for \cgls{ul}, which leads to cross-interference. By contrast, reverse pairing assigns one band to the \cgls{dl} of one system and the \cgls{ul} of the other, and vice versa for the second band, thereby localizing interference more effectively~\cite{Zhang2019}. Notably, \cite{Zhang2019} applies this static pairing framework to \cgls{gso}-\cgls{ngso} coexistence and reports performance improvements in several scenarios. These prior works indicate that interference strongly depends on the operating frequency for both \cgls{dl} and \cgls{ul}. Further, due to the propagation characteristics of satellite communications, using different frequency bands can potentially reduce interference.

\subsection{Contributions}
This paper develops a spectrum-sharing method that introduces a new \cgls{dof} in the \cgls{fdd}-based multi-band satellite systems. Specifically, our contributions are summarized as follows:
\begin{itemize}
    \item We propose dynamic \cgls{fdd}, as a generalization of conventional static \cgls{fdd}, by enabling dynamic band allocation for \cgls{dl} and \cgls{ul}. This introduces an additional \cgls{dof} in the system, and improves interference management capabilities in spectrum-sharing satellite networks.
    
    \item This new \cgls{dof} is exploited in a joint user-satellite assignment and resource allocation framework for multi-satellite systems with overlapping coverage areas. The resulting problem is a highly non-convex mixed-integer optimization problem due to coupling between discrete band selection, scheduling, and continuous power allocation.
    
      \item To solve this problem, we develop a solution algorithm based on alternating optimization and quadratic fractional programing~\cite{Shen2018} techniques. Strong convergence guarantees to first-order optimal points are established. Our proof technique extends the methodological state-of-the-art in convergence analysis for quadratic transform-based algorithms and generalizes well to similar optimization problems.
      
     \item Numerical results demonstrate that the proposed approach achieves up to 30\% improvement in sum-rate in dense user-satellite scenarios. The gains become more pronounced as the number of user and satellite increases, highlighting the effectiveness of dynamic \cgls{fdd} over conventional static configurations.
\end{itemize}

A preliminary version of this work appeared in~\cite{Mukherjee}, where the core idea was introduced. This manuscript provides a more rigorous treatment, including insights on scalability and extensive numerical evaluations, demonstrating the effectiveness of dynamic \cgls{fdd} across diverse network configurations. For reproducibility, the simulations code is available at~\cite{Code2026}.

\subsection{Organization and Notations}
The rest of the paper is organized as follows: in Section~\ref{sec:2}, we discuss the system model for the constellation of satellites serving a common region, then we formulate a joint optimization problem for scheduling and frequency band selection along with transmit powers. Then, we discuss the solution approach using various transformations in Section~\ref{sec:3}. In Section~\ref{sec:4}, we present  numerical results showing the benefits of the proposed dynamic approach. Finally, in Section~\ref{sec:5} conclusions are drawn.
	
Throughout the paper,  bold lowercase and uppercase letters denote vectors and matrices, respectively; scalars use regular font. We consider a vector as a column, and the $\ell_2$-norm and absolute value are denoted by $\lVert \cdot \rVert$ and $\lvert \cdot \rvert$. The sets $\mathbb{R}, \mathbb{R}_+$, and $\mathbb{C}$ denote real, positive real, and complex numbers, respectively. The operation $[a_n]_{n=1}^N$ on a sequence $(a_n)_{n=1}^N$ produces a column vector. Further, the operation $[a_{mn}]_{m=1, n=1}^{m=M,n=N}$ also produces a vector by stacking elements $a_{mn}$, for all $m  = [1, \cdots, M]$ and $n = [1, \cdots, N]$. Transpose and conjugate of a vector are $(\cdot)^T$ and $(\cdot)^*$, respectively. An $m$-dimensional vector with complex entries is denoted as $\mathbb{C}^{m}$. The notation $\sim$ means ``distributed as”, and $\mathcal{CN}(\mu, \sigma^2)$ denotes a circularly symmetric complex Gaussian distribution with mean $\mu$  and  variance $\sigma^2$. Calligraphic uppercase letters, such as, $\mathcal{A}$ represent a set, and $\setminus$ denotes the set difference. For a binary variable ${a} \in \{0,1\}$, its complement $(1-a)$ is denoted as $\bar{a}$, and for a set $\mathcal{A}$, $\mathcal{A}_{-i} = \mathcal{A} \setminus \{i\}$ for $i \in \mathcal{A}$.

\section{System Model and Problem Formulation}
\label{sec:2}
	
Consider a system comprising $J$ \cgls{leo} satellites, each equipped with $N$ antennas, serving a region with a total of $K$ single-antenna \cglspl{ue}, as illustrated in Fig.~\ref{fig:2}. The satellites and \cglspl{ue} are indexed by the sets $\mathcal{J} = \{1, \cdots, J\}$ and $\mathcal{K} = \{1, \cdots, K\}$, respectively. Two frequency bands are available for \cgls{dl} and \cgls{ul} communications, where band $l$, $l \in \{1, 2\}$, has center frequency $\Omega_l$ and bandwidth $B_l$. Satellite $j \in \mathcal{J}$ serves a subset of the \cglspl{ue} on these bands. Each \cgls{ue} maintains separate connections for \cgls{dl} and \cgls{ul}, either to the same or to separate satellites. \cgls{ue} association to satellite~$j$ in the \cgls{dl} is indicated by a binary variable $d_{kj} \in \{0,1\}$, where $d_{kj} = 1$ means \cgls{ue} $k$ is served in \cgls{dl} by satellite~$j$, and $d_{kj} = 0$ otherwise. Similarly, \cgls{ul} association is indicated by a binary variable $u_{kj} \in \{0,1\}$, where $u_{kj} = 1$ means \cgls{ue} $k$ is served in \cgls{ul} by the satellite~$j$, and $u_{kj} = 0$ otherwise. Further, the selection of \cgls{dl} and \cgls{ul} frequency bands from the two available options is treated as a variable, which we refer to as the \textit{spin}. This idea is inspired by~\cite{PetarSpins}, where it was introduced for time-division duplex systems. The spin for satellite $j$ is represented by $r_j \in \{0,1\}$. If $r_j = 1$, band $l = 1$ is used for \cgls{dl} and band $l = 2$ for \cgls{ul}; if $r_j = 0$, the roles of the bands are reversed, as illustrated in Fig.~\ref{fig:3}. This is, essentially, an \cgls{fdd} system at satellite $j$ with improved interference management capabilities through flexible band assignment. Still, each transceiver has to use orthogonal frequency bands for \cgls{ul} and \cgls{dl} transmissions. At \cgls{ue} $k$, this can be enforced by the condition $\sum_{j \in \mathcal{J}} d_{kj} r_{j} =  \sum_{j \in \mathcal{J}} u_{kj} r_{j}$, which is only necessary if the \cgls{ue} is simultaneously served in \cgls{dl} and \cgls{ul}. Thus, for all $k$,
\begin{multline}
	\sum_{j \in \mathcal{J}} d_{kj} = 1 \quad \mathrm{and}\quad \sum_{j \in \mathcal{J}} u_{kj} = 1
	\\\implies
	\sum_{j \in \mathcal{J}} d_{kj} r_{j} = \sum_{j \in \mathcal{J}}  u_{kj} r_{j}.
	\label{eq:1}
\end{multline}
For all other cases, this condition is not relevant.

\begin{figure}[t]
	\centering
	\input{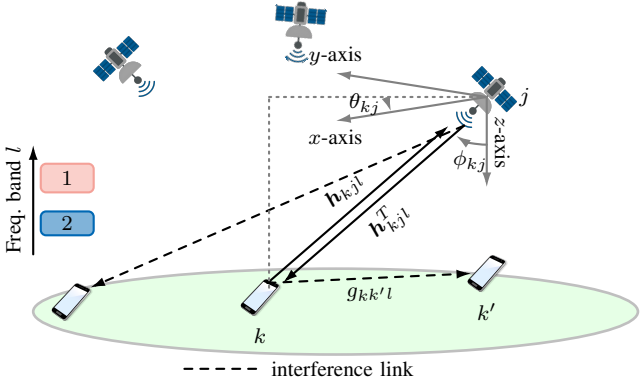}
	\caption{A system of $J$ \cgls{leo} satellites in orbit serving $K$ \cglspl{ue} over   region.}
	\label{fig:2}
\end{figure}
	
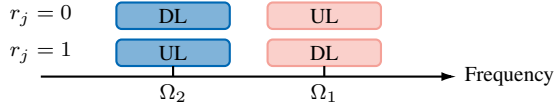
\begin{figure}
	\centering
			{\footnotesize
	\begin{tikzpicture}[yscale = 0.7,every path/.style={>=Latex}] 
		\tikzset{
			smallarrow/.style={
				-{Latex[length=2mm,width=1.2mm]},
				line width=0.3pt
			}
		}
		
		\draw[->, smallarrow,  thick] (1,0) -- (6.5,0) node[right]{Frequency};
		
		\filldraw[fill=myblue!50, draw=myblue, thick, rounded corners=2pt] (2,0.2) rectangle (3.5,0.7) node[midway]{{UL}};
		\filldraw[fill=myred!50, draw=myred, thick, rounded corners=2pt] (4,0.2) rectangle (5.5,0.7) node[midway]{{DL}};
		
		\filldraw[fill=myblue!50, draw=myblue, thick, rounded corners=2pt] (2,0.9) rectangle (3.5,1.4) node[midway]{{DL}};
		\filldraw[fill=myred!50, draw=myred, thick, rounded corners=2pt] (4,0.9) rectangle (5.5,1.4) node[midway]{{UL}};
		
		\node at (1,0.45) {$r_j = 1$};
		\node at (1,1.15) {$r_j = 0$};
		
		\draw[thick] (2.75,0) -- (2.75,0.2);
		\draw[ thick] (4.75,0) -- (4.75,0.2);
		
		\node[below] at (2.75,0) {$\Omega_2$};
		\node[below] at (4.75,0) {$\Omega_1$};
		
	\end{tikzpicture}
}
	\caption{Proposed dynamic \cgls{fdd} framework, where the two available frequency bands at each satellite can be flexibly assigned to either \cgls{dl} or \cgls{ul}. This additional \cgls{dof} is captured through a binary variable, spin $r_j$ for satellite $j$, whose value determines the corresponding band-direction mapping.}
	\label{fig:3}
\end{figure}
	
We denote the band-$l$ channel from satellite $j$ to \cgls{ue}~$k$ as $\bm h_{kjl} \in \mathbb{C}^{N}$. This channel has a dominant \cgls{los} component with only small non-\cgls{los} effects in the ground-segment. The \cglspl{ue} are assumed to know their own position and those of all satellites. In particular, let $m_{kj}$, $\theta_{kj}$, and $\phi_{kj}$ be the distance, azimuth, and elevation angle from satellite $j$ to \cgls{ue} $k$. Then, the \cgls{los} component of 
\begin{align}
	\bm h_{kjl} = 	\sqrt{\beta_{kjl}}\, \bm b_l(\theta_{kj}, \phi_{kj}),
	\label{eq:2}
\end{align}
with path loss $\beta_{kjl} = \left( {c}/{4 \pi m_{kj} \Omega_l}\right)^2$ and array-response vector \cite{Roper2024} is
\begin{equation}
	\bm b_l = 
	\left[ \exp\left(-\frac{\sqrt{-1} \,2\pi \Omega_l}{c}\left( D^x_n \psi^x + D^y_n \psi^y\right)\right) \right]_{n=1}^N,
	\label{eq:3}
\end{equation}
where $\psi^x = \cos{( \phi_{kj} )}\cos{(\theta_{{kj}})}, \psi^y = \cos{( \phi_{kj} )}\sin{(\theta_{kj})}$, $D^x_n, D_n^y$ denote the $x$ and $y$-coordinate, respectively of $n$-th antenna element, and  $c$ is the speed-of-light. We model $\bm h_{kjl}$ as purely \cgls{los} but note that statistical \cgls{csi} is straightforward to incorporate into $\beta_{kjl}$. Further, we assume negligible interference between satellites due to the high directivity and downward positioning of the satellite antenna array. However, we need to take inter-\cgls{ue} interference into account. Since accurate and instantaneous \cgls{ue}-to-\cgls{ue} \cgls{csi} is infeasible to obtain, we model the band-$l$ channel $g_{kk'l} = g_{k'kl}$ from \cgls{ue} $k$ to $k'$ based on position-dependent long-term statistics.

Now, consider that \cgls{ue}~$k$ is served by satellite~$j$ in the \cgls{dl}, indicated by the binary variable $d_{kj} = 1$. The effective \cgls{dl} interference channel from satellite~$j' \in \mathcal{J}$ to $k$, while serving another \cgls{ue}~$k' \in \mathcal{K}_{-k}$ in the \cgls{dl} (i.e., $d_{k'j'}=1$), is given by
\begin{align}
	\left(\bm \gamma^{\mathrm{dl}}_{kjj'} \right)^T &= \overline{\lvert r_j - r_{j'} \rvert} \,  \Big[ r_{j'} \, \bm h_{kj',l=1}^T  + \overline{r_{j'}} \, \bm h_{kj',l=2}^T \Big], 
	\label{eq:4}
\end{align}
where $r_j$ and $r_{j'}$ denote the spins of satellites~$j$ and~$j'$, respectively, and $\lvert r_j - r_{j'} \rvert$ quantifies their relative spin. Interference in the \cgls{dl} arises from another satellite~$j'$ only when both satellites operate with the same spin, which is explicitly captured in the above formulation. Moreover, $\bm \gamma^{\mathrm{dl}}_{kjj'}$ also accounts for the direct \cgls{dl} channel between satellite~$j$ and \cgls{ue}~$k$ in the special case $j' = j$, since $\overline{\lvert r_j - r_j \rvert} = 1$. 
Similarly, if \cgls{ue}~$k$ is served by satellite~$j$ in the \cgls{ul}, indicated by $u_{kj} = 1$, then the effective \cgls{ul} interference channel from \cgls{ue}~$k' \in \mathcal{K}$ to satellite~$j$, where $k'$ is served by satellite~$j' \in \mathcal{J}$ in the \cgls{ul} (i.e., $u_{k'j'} = 1$), is given by
\begin{align}
\bm \gamma^{\mathrm{ul}}_{k'jj'} &= \overline{\lvert r_j - r_{j'} \rvert} \, \Big[ r_{j'} \, \bm h_{k'j,l=2} + \overline{r_{j'}} \, \bm h_{k'j,l=1} \Big].
	\label{eq:5}
\end{align}
Note that $\bm \gamma^{\mathrm{ul}}_{k'jj'}$ also captures the direct \cgls{ul} channel from \cgls{ue}~$k$ to satellite~$j$ when $k' = k$ and $j' = j$. This interference arises when both satellites operate with the same spin. 
Additionally, a \cgls{ue} causes interference to another \cgls{ue} during \cgls{ul} transmissions, if the spin of the associated satellites $j$ and $j'$ differs. Suppose \cgls{ue}~$k$ is served in the \cgls{dl} by satellite~$j$ (i.e., $d_{kj} = 1$), while \cgls{ue}~$k' \in \mathcal{K}_{-k}$ is served in the \cgls{ul} by satellite~$j' \in \mathcal{J}_{-j}$ (i.e., $u_{k'j'} = 1$). Then, the effective interference channel from \cgls{ue}~$k'$ to \cgls{ue}~$k$ is given by
\begin{align}
	\nu_{kk'jj'} &= \lvert r_j - r_{j'} \rvert \, \Big[ r_{j'} \, g_{k'k,l=1} + \overline{r_{j'}} \, g_{k'k,l=2} \Big],
	\label{eq:6}
\end{align}
where the effective channels $\bm{\gamma}^{\mathrm{dl}}_{kjj'}, \bm{\gamma}^{\mathrm{ul}}_{k'jj'},  \text{ and } \nu_{kk'jj'}$ are selected by the spin variables $r_j$ and $r_{j'}$.

Note that we assume the positions of all \cglspl{ue} are available at the satellites. This is considering \gls{gnss}-capable devices that are compliant with 3GPP standards~\cite{3gpp}. Given accurate position information, the satellites can reliably estimate the \cgls{ue}-satellite channel based on a \cgls{los} propagation model, which is a reasonable assumption for LEO satellite links. Obtaining \cgls{ue}-\cgls{ue}  \cgls{csi} at the satellite is challenging due to signaling and latency constraints. Therefore, we utilize the \cgls{ue} positions to model these channels as \cgls{los}. This corresponds to a conservative worst-case assumption, as it captures the strongest plausible interference conditions, enabling the proposed optimization framework to be implemented in a practical manner.
	
	\begin{figure*}[b]
		\normalsize
		\setcounter{MYtempeqncnt}{\value{equation}}
		\setcounter{equation}{9}
		\hrulefill 
		\vspace{0.1cm}
		\begin{minipage}{\textwidth}
			\begin{align}
				f_0&(\bm d, \bm r, \bm u, \bm p^{\text{dl}}, \bm p^{\text{ul}}) =   \sum_{k \in \mathcal{K}} \sum_{j \in\mathcal{J}} \Big[ \log_2 \Big(1 +  u_{kj} \, p_{k}^{\text{ul}}  \left\lVert \bm \gamma^{\mathrm{ul}}_{kjj} \right\rVert^2 	 \Big/  \Big( \sigma^2 +  \sum_{j' \in \mathcal{J}} \sum_{k' \in \mathcal{K}_{-k}} u_{k'j'} \, p_{k'}^{\text{ul}} \, \left\lvert \bm v_{kj}^T \bm \gamma^{\mathrm{ul}}_{k'jj'}\right\rvert^2   \Big) \Big)	\nonumber\\
				&+ \log_2 \Big( 1 + d_{kj}\, p_{k}^{\text{dl}} \, \left\lVert \bm \gamma^{\mathrm{dl}}_{kjj}  \right\rVert^2\Big/ \Big( \sigma^2 +  \sum_{j' \in \mathcal{J}} \sum_{k' \in \mathcal{K}_{-k}} d_{k'j'}\, p_{k'}^{\text{dl}} \, \left\lvert ( \bm \gamma^{\mathrm{dl}}_{kjj'} )^T  \bm w_{k'j'}\right\rvert^2    + \sum_{j' \in \mathcal{J}_{-j} }\sum_{k' \in \mathcal{K}_{-k}}  u_{k'j'} \, p_{k'}^{\text{ul}}\, \lvert \nu_{kk'jj'}\rvert^2\Big) \Big) \Big].  \label{eq:10} \\
				f_1&(\bm d, \bm r, \bm u, \bm p^{\text{dl}}, \bm p^{\text{ul}}, \,\bm \chi^{\text{dl}} , \bm \chi^{\text{ul}}  ) = \nonumber \\
				&  \sum_{k \in \mathcal{K}} \sum_{j\in \mathcal{J}} \Big[ \log_2 \left(1+  	\chi^{\text{dl}}_{kj} \right)	+ \log_2 \left( 1 + \chi^{\text{ul}}_{kj} \right) - \chi^{\text{dl}}_{kj}  - \chi^{\text{ul}}_{kj}  + \left(1+  \chi^{\text{ul}}_{kj} \right) u_{kj} \, p_{k}^{\text{ul}}  \left\lVert \bm \gamma^{\text{ul}}_{kjj}  \right\rVert^2 	 \Big/  \Big( \sigma^2 +  \sum_{j' \in \mathcal{J}} \sum_{k' \in \mathcal{K}} u_{k'j'} \, p_{k'}^{\text{ul}}  \left\lvert \bm v^T_{kj} \bm \gamma^{\text{ul}}_{k'jj'} \right\rvert^2   \Big) \nonumber \\
				&\hspace{15pt}+ \left(1+  \chi^{\text{dl}}_{kj} \right) d_{kj\, } p_{k}^{\text{dl}} \, \left\lVert\bm \gamma^{\text{dl}}_{kjj} \right\rVert^2 \Big/ \Big( \sigma^2 +  \sum_{j' \in \mathcal{J}} \sum_{k' \in \mathcal{K}} d_{k'j'\, } p_{k'}^{\text{dl}} \, \left\lvert (\bm \gamma^{\text{dl}}_{kjj'})^T \bm w_{k'j'}\right\rvert^2    + \sum_{j' \in \mathcal{J}_{-j}} \sum_{k' \in \mathcal{K}_{-k} }  u_{k'j'} \, p_{k'}^{\text{ul}} \,\lvert \nu_{kk'jj'}\rvert^2\Big)\Big]. \label{eq:11} \\
				f_2&(\bm d, \bm r, \bm u, \bm p^{\text{dl}}, \bm p^{\text{ul}}, \,\bm \chi^{\text{dl}} , \bm \chi^{\text{ul}}, \bm \xi^{\text{dl}}, \bm \xi^{\text{ul}} ) = \nonumber\\
				&  \sum_{j\in \mathcal{J}} \sum_{k \in \mathcal{K}} \Big[\log_2 \left(1+  	\chi^{\text{dl}}_{kj} \right)	+ \log_2 \left( 1 + \chi^{\text{ul}}_{kj} \right) - \chi^{\text{dl}}_{kj}  - \chi^{\text{ul}}_{kj} 
				+  2 \xi_{kj}^{\text{dl}}  d_{kj }\sqrt{(1+  	\chi^{\text{dl}}_{kj}) p_{k}^{\text{dl}}}  \left\lVert\bm \gamma^{\text{dl}}_{kjj}\right\rVert + 2 \xi_{kj}^{\text{ul}} \, u_{kj }\sqrt{(1+  	\chi^{\text{ul}}_{kj} ) p_{k}^{\text{ul}}} \, \left\lVert\bm \gamma^{\text{ul}}_{kjj}  \right\rVert   \nonumber \\
				& \hspace{50pt}-   \left(\xi_{kj}^{\text{dl}}\right)^2 \Big(\sigma^2 +  \sum_{j'\in \mathcal{J}} \sum_{k' \in \mathcal{K}}  d_{k'j'\, } p_{k'}^{\text{dl}} \, \left\lvert (\bm \gamma^{\text{dl}}_{kjj'})^T \bm w_{k'j'} \right\rvert^2  +  \sum_{j' \in \mathcal{J}_{-j}} \sum_{k' \in \mathcal{K}_{-k}}  u_{k'j'} \, p_{k'}^{\text{ul}} \,\lvert \nu_{kk'jj'} \rvert^2 \Big) \nonumber \\
				& \hspace{150pt}-   \left(\xi_{kj}^{\text{ul}} \right)^2 \Big( \sigma^2 +   {\sum_{j' \in \mathcal{J}} \sum_{k' \in \mathcal{K}} u_{k'j'}  p_{k'}^{\text{ul}}  \left\lvert \bm v^T_{kj} \bm \gamma^{\text{ul}}_{k'jj'} \right\rvert^2} \Big)\Big]. \label{eq:12} 
			\end{align} 
			
		\end{minipage}
		\setcounter{equation}{\value{MYtempeqncnt}}
		\hrulefill
		\vspace*{4pt}
	\end{figure*}

The received signals at \cgls{ue}~$k$ in the \cgls{dl} and at satellite~$j$ in the \cgls{ul}, corresponding to link~$kj$, are
\begin{subequations}\label{eq:7}
\begin{align}
	y_{kj}^{\text{dl}} =&\, d_{kj} \, \sqrt{p_{k}^{\text{dl}}} \, \lVert \bm \gamma^{\mathrm{dl}}_{kjj} \rVert \, s^{\text{dl}}_{k}  \nonumber\\
	&+ \sum_{j' \in \mathcal{J}} \sum_{k' \in \mathcal{K}_{-k}} d_{k'j'} \, \sqrt{p_{k'}^{\text{dl}}} \, \left( \bm \gamma^{\mathrm{dl}}_{kjj'} \right)^T  \bm w_{k'j'} \, s^{\text{dl}}_{k'} \nonumber\\
	&+ \sum_{j' \in \mathcal{J}_{-j}} \sum_{k' \in \mathcal{K}_{-k}} u_{k'j'} \, \sqrt{p_{k'}^{\text{ul}}} \, \nu_{kk'jj'} \, s^{\text{ul}}_{k'} + n^{\mathrm{dl}}_{kj},
	\label{eq:7a} \\
	y^{\text{ul}}_{kj} =&\, u_{kj} \sqrt{p_{k}^{\text{ul}}} \, \lVert \bm \gamma^{\mathrm{ul}}_{kjj} \rVert \, s^{\text{ul}}_{k} \nonumber \\
	&+ \sum_{j' \in \mathcal{J}} \sum_{k' \in \mathcal{K}_{-k}} u_{k'j'} \, \sqrt{p_{k'}^{\text{ul}}} \, \bm v_{kj}^T \bm \gamma^{\mathrm{ul}}_{k'jj'} \, s^{\text{ul}}_{k'} + n^{\mathrm{ul}}_{kj},
	\label{eq:7b}
\end{align}
\end{subequations}
respectively, where  $s_{k'}^{\mathrm{dl}}$ denotes the unit-energy symbol transmitted in the \cgls{dl} towards \cgls{ue}~$k'$, $s_{k'}^{\mathrm{ul}}$ represents the unit-energy symbol transmitted by \cgls{ue}~$k'$ in the \cgls{ul}. The receiver noise terms $n_{kj}^{\mathrm{dl}}$ for the \cgls{dl} and $n_{kj}^{\mathrm{ul}}$ for the \cgls{ul} are modeled as circularly-symmetric complex Gaussian random variables with variance $\sigma^2$, i.e., $n_{kj}^{\mathrm{dl}}, n_{kj}^{\mathrm{ul}} \sim \mathcal{CN}(0, \sigma^2)$. The transmit power allocated to user $k'$ in the \cgls{dl} and \cgls{ul} is denoted by $p_{k'}^{\mathrm{dl}}$ and $p_{k'}^{\mathrm{ul}}$, respectively.
Further, $\bm w_{k'j'} \in \mathbb{C}^{N }$ denotes the unit-norm \cgls{dl} \gls{mrt} precoder for link $k'j'$, and the vector $\bm v_{kj} \in \mathbb{C}^{N}$ is the unit-norm \gls{mrc}  at satellite~$j$ for \cgls{ue}~$k$ in the \cgls{ul}. In particular, 
\begin{subequations}\label{eq:8}
\begin{align}
	\bm w_{kj} &= \begin{cases}
		\bm h_{kj, l=1}^* \Big/\lVert\bm h_{kj, l=1} \rVert,	 & \text{if } r_j = 1 \\
		\bm h_{kj, l=2}^* \Big/\lVert\bm h_{kj, l=2} \rVert,	 & \text{if } r_j = 0
	\end{cases} \label{eq:8a}\\
	\bm v_{kj} &= \begin{cases}
		\bm h_{kj, l=2}^* \Big/\lVert\bm h_{kj, l=2} \rVert,	 & \text{if } r_j = 1 \\
		\bm h_{kj, l=1}^* \Big/\lVert\bm h_{kj, l=1} \rVert.	 & \text{if } r_j = 0
	\end{cases} \label{eq:8b}							
\end{align}
\end{subequations}
With this system model, we develop the corresponding resource allocation problem for dynamic \cgls{fdd} next.

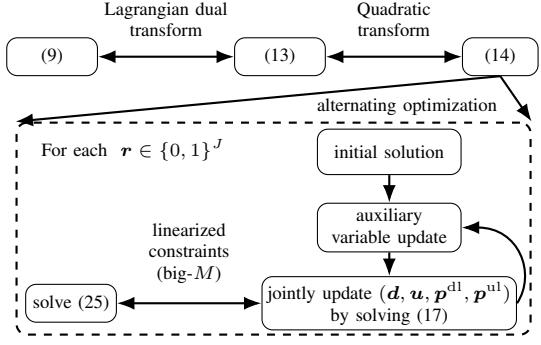
\begin{figure}[t]
	\centering
	\usetikzlibrary{arrows.meta, positioning, fit}	
\tikzset{
	problem/.style={rectangle, draw, rounded corners,
		minimum width=1.2cm, minimum height=0.5cm,
		align=center},
	stepbox/.style={rectangle, draw, rounded corners,
		align=center, inner sep=2pt,
		minimum width=1.9cm, minimum height=0.6cm},
	arrow/.style={-{Latex}, thick},
	doublearrow/.style={<->, thick, >=Latex
	},
	bigbox/.style={draw, rounded corners, thick,
		inner sep=0.3cm, dashed}
}

{\scriptsize
	\begin{tikzpicture}[xscale=1.8, node distance=1.8cm]
		
		\node[problem] (P0) {\eqref{eq:9}};
		\node[problem, right=of P0] (P1) {\eqref{eq:13}};
		\node[problem, right=of P1, minimum width=1cm] (P2) {\eqref{eq:14}};
		
		\draw[doublearrow, line cap=round]
		(P0) -- (P1)
		node[midway, above=5pt, sloped, align=center]
		{Lagrangian dual \\ transform};
		
		\draw[doublearrow, line cap=round]
		(P1) -- (P2)
		node[midway, above=5pt, sloped, align=center]
		{Quadratic \\ transform };
		
		\node[bigbox, below=.6cm of P2, xshift=-3cm,
		minimum width=6.8cm, minimum height=2.8cm] (iterBox) {};
		
		\node[above left=.2cm and -3.9cm of iterBox.north west, anchor=west]
		{alternating optimization};
		
		\draw[->, thick, >=Latex] (P2.south) -- (iterBox.north west);
		\draw[->, thick, >=Latex] (P2.south) -- (iterBox.north east);
		
		\node[stepbox] (box1) at ($(iterBox.north)+(.85,-0.4cm)$)
		{initial solution};
		\node[stepbox] (box2) at ($(iterBox.north)+(.85,-1.4cm)$)
		{auxiliary \\ variable update};
		\node[stepbox] (box3) at ($(iterBox.north)+(.85,-2.4cm)$)
		{jointly update $(\bm d, \bm u, \bm p^{\mathrm{dl}},  \bm p^{\mathrm{ul}})$\\   by solving \eqref{eq:17}};
		
		\draw[->, thick, >=Latex] (box1.south) -- (box2.north);
		\draw[->, thick, >=Latex] (box2.south) -- (box3.north);
		
		\draw[->, thick, >=Latex]
		(box3.east) to[out=70, in=0] (box2.east);
		
		\node[problem] (box4) at ($(box3.west)+(-1.4,0)$) {solve \eqref{eq:25}};
		
		\draw[<->, thick, >=Latex] (box3.west) -- (box4.east) 	node[midway, above=5pt, sloped, align=center]
		{linearized\\constraints\\(big-$M$)};
		
		\node[anchor=north west, align=left]
		at ([xshift=4pt,yshift=-4pt]iterBox.north west)
		{\text{For each } $\boldsymbol r \in \{0,1\}^J$};
		
	\end{tikzpicture}
}	
	\caption{Illustration of sequence of reformulations to solve \eqref{eq:9}.}
	\label{fig:10}
\end{figure}

\subsection{Resource Allocation for Dynamic FDD}
\label{sec:2a}

To efficiently allocate radio resources in the system, we formulate a joint optimization problem that considers user association, power control, and flexible band assignment (spin) across satellites. The ultimate goal is to enhance spectral efficiency while respecting practical system constraints. Specifically, we aim to maximize the overall two-way communication rate between satellites and \cglspl{ue}, which captures the sum of both \cgls{dl} and \cgls{ul} rates. In particular, we solve
\begin{subequations}\label{eq:9}
	\begin{align}
		\underset{\bm d, \bm r, \bm u, \bm p^{\text{dl}}, \bm p^{\text{ul}}}{\text{max}} &	\, \, \, \, f_0(\bm d, \bm r, \bm u, \bm p^{\text{dl}}, \bm p^{\text{ul}}), \label{eq:9a}\\
		\text{s.t. \hspace{1pt}	}
		& \text{FDD condition in }\eqref{eq:1}, \text{ for all } k,  \label{eq:9b} \\
		& \sum_{j \in \mathcal{J}} d_{kj} \leq 1, \sum_{j \in \mathcal{J}} u_{kj} \leq 1,  \text{ for all } k,\label{eq:9c}\\
		& \sum_{k \in \mathcal{K}} d_{kj} \, p_{k}^{\text{dl}} \leq p_{j}^{\text{max}}, \text{ for all } j,\label{eq:9d} \\
		& p_{k}^{\text{ul}} \leq p_{k}^{\text{max}},  \text{ for all } k, \label{eq:9e} \\
		& p_{k}^{\text{dl}} \geq 0, \, p_{k}^{\text{ul}} \geq 0, \text{ for all } k, \label{eq:9f} \\
		& d_{kj}, r_{j}, u_{kj} \in \{0,1\},  \text{ for all } k \text{ and } j.\label{eq:9g} 
	\end{align}
\end{subequations}
Here, the variables are grouped as $\bm d = [d_{kj}]_{k=1, j=1}^{k= K, j = J}$,  $\bm u = [u_{kj}]_{k=1, j=1}^{k= K, j = J}$, $\bm r = [r_j]_{j=1}^J$, $\bm p^{\text{dl}} = [p_{k}^{\text{dl}}]_{k=1}^K$, and $\bm p^{\text{ul}} = [p_{k}^{\text{ul}}]_{k=1}^K$.
The objective $f_0$ is given in~\eqref{eq:10} at the bottom of the page.  Constraint \eqref{eq:9b} captures the fact that each frequency band can be used for either transmission or reception but not both. Constraint~\eqref{eq:9c} restricts each UE to connect
with at most one satellite in DL and at most one in UL. Power constraint \eqref{eq:9d} limits the total transmit power used by satellite~$j$ for serving its scheduled \cglspl{ue} in the \cgls{dl} to stay within its maximum power budget $p_j^{\text{max}}$.
Likewise, constraint \eqref{eq:9e} restricts each \cgls{ue}~$k$ from exceeding its own transmit power limit $p_k^{\text{max}}$ in the \cgls{ul}. The constraint \eqref{eq:9f} ensures all power values are non-negative, and \eqref{eq:9g} enforces the binary domain of the scheduling and spin variables.

\section{Solution to the Mathematical Program \eqref{eq:9}}
\label{sec:3}

The optimization problem in \eqref{eq:9} is highly non-convex due to the interference-limited sum-rate expression and the binary decision variables~\cite{Luo2008}. Notably, even if the binary variables are relaxed to continuous surrogates, the problem remains non-convex due to the fractional and interference-coupled structure of the \cgls{sinr}, and known to be NP-hard~\cite{Luo2008}. As a result, computing a global optimum with reasonable computational effort is intractable for practically relevant system dimensions. Instead of pursuing global optimality at prohibitive complexity, our objective is to design a structured solution approach that guarantees feasibility and achieves high-quality performance with manageable computational cost. To this end, we reformulate the original problem into a sequence of more tractable problems. This reformulation is carefully constructed to exploit existing, well-established \gls{minlp} solvers for specific sub-structures of the problem. By embedding these solvers within our algorithmic framework, we effectively reduce computational burden while maintaining strong solution quality.

We apply a sequence of transformations, summarized in Fig.~\ref{fig:10}, to obtain a structure suitable for alternating optimization. First, the Lagrangian dual transform \cite{Shen2018} is applied to decouple the logarithmic rate expressions from the fractional \cgls{sinr} terms through the introduction of auxiliary variables. This reformulation moves the \cgls{sinr} terms outside the logarithms and removes the direct coupling between logarithmic and fractional components in the objective. For fixed values of the remaining variables, the resulting objective becomes strictly convex with respect to the introduced auxiliary variables, allowing their unique and optimal update within an alternating optimization framework. Subsequently, the quadratic transform \cite{Shen2018} is applied to the fractional components. By introducing additional auxiliary variables, each term is converted into an equivalent separable quadratic form. Similar to the previous step, the transformed objective is strictly convex in these auxiliary variables when other variables are fixed, which guarantees a unique optimal update. These properties collectively yield a block-structured formulation in which each auxiliary block can be solved to global optimality. The remaining non-convexity arises from the discrete vector $\bm r$. To address this, we perform an exhaustive search over all feasible realizations of $\bm r$. For a fixed $\bm r$, the joint update of $(\bm d, \bm u, \bm p^{\mathrm{dl}}, \bm p^{\mathrm{ul}})$ is formulated as a \gls{minlp}. The formulated problem contains disjunctive constraints due to the presence of binary-continuous products, and we use the Big-$M$ method \cite{Glover1975, Dejans2025}, to handle those. This enables the use of standard \cgls{minlp} solvers for the joint update. By combining convex block-wise updates with solver-based mixed-integer optimization, the proposed framework maintains computational tractability while achieving high-quality feasible solutions.

\subsection{Equivalence Transforms}

First, we adapt the Lagrangian dual transform from Theorem~3 of \cite{Shen2018} to transform the problem~$\eqref{eq:9}$ into
	\setcounter{MYtempeqncnt}{\value{equation}}
	\setcounter{equation}{12}
\begin{subequations}\label{eq:13}
	\begin{align}
		\underset{
			\substack{
				\bm d, \bm r, \bm u, \bm p^{\text{dl}}, \bm p^{\text{ul}},
				\bm \chi^{\text{dl}}, \bm \chi^{\text{ul}}
			}
		}{\text{max}} &	\, \, \, \, f_1\left(\bm d, \bm r, \bm u, \bm p^{\text{dl}}, \bm p^{\text{ul}}, \bm \chi^{\text{dl}} , \bm \chi^{\text{ul}}  \right) \label{eq:13a} \\
		\text{s.t. }&
		\eqref{eq:9b}-\eqref{eq:9g}, \\
		& \chi_{kj}^{\mathrm{dl}}\in \mathbb{R}_+ , \chi_{kj}^{\mathrm{ul}} \in \mathbb{R}_+
	\end{align}
\end{subequations}
where we introduce auxiliary variables $\bm \chi^{\text{dl}} = [\chi^{\text{dl}}_{kj}]_{k=1, j=1}^{k= K, j = J}$ for the \cgls{dl}, and $\bm \chi^{\text{ul}} = [\chi^{\text{ul}}_{kj}]_{k=1, j=1}^{k= K, j = J}$ for the \cgls{ul}. The objective $f_1$ is shown at  bottom of the previous page in \eqref{eq:11}.

Next, we adapt the quadratic transform from Theorem~1 of \cite{Shen2018} to convert the ratios of sum in $\eqref{eq:11}$ to the sum of terms, obtaining
\begin{subequations}\label{eq:14}
	\begin{align}
		\underset{ \substack{\bm d, \bm r, \bm u, \bm p^{\text{dl}}, \bm p^{\text{ul}}, \\ \bm \chi^{\text{dl}} , \bm \chi^{\text{ul}}, \bm \xi^{\text{dl}}, \bm \xi^{\text{ul}} }}{\text{max}}&	\, \, \, \, f_2\left(\bm d, \bm r, \bm u, \bm p^{\text{dl}}, \bm p^{\text{ul}}, \bm \chi^{\text{dl}} , \bm \chi^{\text{ul}}, \bm \xi^{\text{dl}}, \bm \xi^{\text{ul}} \right) \label{eq:12a} \\
		\text{s.t. }&
		\eqref{eq:9b}-\eqref{eq:9g},  \\
		&   \chi_{kj}^{\mathrm{dl}}, \chi_{kj}^{\mathrm{ul}}  \in \mathbb{R}_+,\quad 
		\xi_{kj}^{\mathrm{dl}}, \xi_{kj}^{\mathrm{ul}}  \in \mathbb{R},
	\end{align}
\end{subequations}
where the auxiliary variables are $\bm \xi^{\text{dl}} = [\xi^{\text{dl}}_{kj}]_{k=1, j=1}^{k= K, j = J}$ and $\bm \xi^{\text{ul}} = [\xi^{\text{ul}}_{kj}]_{k=1, j=1}^{k= K, j = J}$. Further, the objective $f_2$ is shown at the bottom of previous page in \eqref{eq:12}. The equivalence between the problems $\eqref{eq:9}$, $\eqref{eq:13}$, and $\eqref{eq:14}$ is established in the following proposition:
\begin{proposition}
	\label{prop:1}
		    Let $(
				\bm d^\star, \bm r^\star, 
				\bm u^\star, 
				\bm p^{\mathrm{dl}\star}, \bm p^{\mathrm{ul}\star}, \;
				\bm \chi^{\mathrm{dl}\star}, \bm \chi^{\mathrm{ul}\star}, \bm \xi^{\mathrm{dl}\star},\allowbreak \bm \xi^{\mathrm{ul}\star} 
				)$
  be a solution to \eqref{eq:14}. Then, $(
		\bm d^\star,\allowbreak \bm r^\star, \allowbreak
		\bm u^\star, \allowbreak
		\bm p^{\mathrm{dl}\star},\allowbreak \bm p^{\mathrm{ul}\star},\allowbreak \;
		\bm \chi^{\mathrm{dl}\star}, \allowbreak \bm \chi^{\mathrm{ul}\star} 
		)$
		solves \eqref{eq:13}. Moreover, if $(
			\bm d^\dagger, \bm r^\dagger, \allowbreak
			\bm u^\dagger, \allowbreak
			\bm p^{\mathrm{dl}\dagger}, \allowbreak\bm p^{\mathrm{ul}\dagger}, \;
			\bm \chi^{\mathrm{dl}\dagger},\allowbreak \bm \chi^{\mathrm{ul}\dagger} 
			)$
		is a solution of \eqref{eq:13}, then $(
		\bm d^\dagger, \bm r^\dagger, 
		\bm u^\dagger, 
		\bm p^{\mathrm{dl}\dagger}, \bm p^{\mathrm{ul}\dagger}
		)$
		solves \eqref{eq:9}.
\end{proposition}

\begin{IEEEproof}
	See Appendix~\ref{appendix:1}.
\end{IEEEproof}

To address the original problem in~\eqref{eq:9}, we begin by solving its equivalent reformulation in~\eqref{eq:14}. As established in the proposition above, the optimal solution of~\eqref{eq:14} is also optimal for the main problem~\eqref{eq:9}, which allows us to focus on~\eqref{eq:14} without any loss of optimality. The objective function $f_2$ is non-convex due to the multiplicative coupling among scheduling variables, transmit powers, and auxiliary variables. In particular, the terms $d_{kj}\sqrt{(1+\chi^{\text{dl}}_{kj})p_k^{\text{dl}}}$ and  $u_{kj}\sqrt{(1+\chi^{\text{ul}}_{kj})p_k^{\text{ul}}}$ couple binary scheduling variables, power variables, and auxiliary variables. Moreover, the quadratic terms $(\xi_{kj}^{\text{dl}})^2$ and $(\xi_{kj}^{\text{ul}})^2$ are multiplied by interference expressions that depend on all \cglspl{ue} and satellites. Therefore, $f_2$ is not jointly concave with respect to the optimization variables.	However, when other variables are fixed, $f_2$ becomes strictly concave in  $ \xi^{\text{dl}}_{kj}$ and in $\xi^{\text{ul}}_{kj}$. Further, keeping other variables fixed, $f_1$ is strictly concave in $ \chi^{\text{dl}}_{kj}$ and in $\chi^{\text{ul}}_{kj}$. This structure motivates an alternating optimization approach, which generates candidate solutions that correspond to local maxima of problem \eqref{eq:9}.

To solve~\eqref{eq:14}, the optimization variables are updated sequentially in three blocks. Specifically, for a fixed scheduling vector $\bm r$, we first update the auxiliary variables $\bm \chi^{\mathrm{dl}}$ and $\bm \chi^{\mathrm{ul}}$; next, the auxiliary variables $\bm \xi^{\mathrm{dl}}$ and $\bm \xi^{\mathrm{ul}}$ are updated; and finally, the main continuous variables $\bm d$, $\bm u$, $\bm p^{\mathrm{dl}}$, and $\bm p^{\mathrm{ul}}$ are optimized jointly. At each stage, one block of variables is optimized while the remaining blocks are held fixed, which guarantees a monotonic improvement of the objective value. This helps us establish the convergence to a stationary point \cgls{wrt} the continuous variables for fixed binary variables. In the following subsections, we detail each of these updates and explain how their sequential execution yields a solution to~\eqref{eq:14}, and consequently to the original problem in~\eqref{eq:9}.

\begin{figure*}[b]
	\normalsize
	\setcounter{MYtempeqncnt}{\value{equation}}
	\setcounter{equation}{17}
	\hrulefill 
	\vspace{0.5em}
	\begin{minipage}{\textwidth}
		\begin{align}
		f_3&(\bm d, \bm r, \bm u, \bm p^{\text{dl}}, \bm p^{\text{ul}}, \,\bm \chi^{\text{dl}} , \bm \chi^{\text{ul}}, \bm \xi^{\text{dl}}, \bm \xi^{\text{ul}} ) 
		=   \sum_{k \in \mathcal{K}} \sum_{j\in \mathcal{J}}  \Big[\log_2 (1+  	\chi^{\text{dl}}_{kj} )	+ \log_2 ( 1 + \chi^{\text{ul}}_{kj} ) - \chi^{\text{dl}}_{kj}  - \chi^{\text{ul}}_{kj} - (\xi_{kj}^{\text{dl}} \sigma)^2 - (\xi_{kj}^{\text{ul}} \sigma)^2
			\nonumber \\
			&+  2 \xi_{kj}^{\text{dl}} \, d_{kj} \sqrt{(1+  	\chi^{\text{dl}}_{kj} ) p_k^{\mathrm{dl}}} \, \lVert\bm \gamma^{\text{dl}}_{kjj}\rVert + 2 \xi_{kj}^{\text{ul}} \, u_{kj}\sqrt{(1+  	\chi^{\text{ul}}_{kj} ) p_k^{\mathrm{ul}}} \, \lVert\bm \gamma^{\text{ul}}_{kjj}  \rVert -d_{kj} p_k^{\mathrm{dl}}  \sum_{j' \in \mathcal{J}} \sum_{k' \in \mathcal{K}} (\xi_{k'j'}^{\text{dl}})^2  \, \lvert (\bm \gamma^{\text{dl}}_{k'j'j})^T \bm w_{kj} \rvert^2    \nonumber \\ 
			& - u_{kj} p_k^{\mathrm{ul}} \sum_{j' \in \mathcal{J}_{-j}} \sum_{k' \in \mathcal{K}_{-k}} (\xi_{k'j'}^{\text{dl}})^2  \,\lvert \nu_{k'kj'j} \rvert^2- u_{kj} p_k^{\mathrm{ul}}  \sum_{j' \in \mathcal{J}} \sum_{k' \in \mathcal{K}} (\xi_{k'j'}^{\text{ul}})^2   \lvert \bm v^T_{k'j'} \bm \gamma^{\text{ul}}_{kj'j} \rvert^2   \Big]. \label{eq:18}
		\end{align}
	\end{minipage}
	\setcounter{equation}{\value{MYtempeqncnt}}
\end{figure*}

\subsection{Block  Updates}	
\subsubsection{Update  {$\bm \chi^{\mathrm{dl}}$ and $ \bm \chi^{\mathrm{ul}}$}}
\label{sec:3a}
The objective $f_1$ is strictly concave in $\chi^{\text{dl}}_{kj}$ and $\chi^{\text{ul}}_{kj}$. Hence, by keeping other parameters fixed, the unique optimal value of $\chi^{\text{dl}}_{kj}$ and $\chi^{\text{ul}}_{kj}$ can be obtained from $\partial f_1 / \partial \chi^{\text{dl}}_{kj} = 0$ and $\partial f_1 / \partial \chi^{\text{ul}}_{kj} = 0$, respectively, as
\begin{subequations}\label{eq:15}
	\begin{align}
		\chi_{kj}^{\text{dl}\star} &=\, d_{kj}\, p_{k}^{\text{dl}} \, \lVert \bm \gamma^{\mathrm{dl}}_{kjj} \rVert^2 \nonumber \\ \Big/ &\Big( \sigma^2 +  \sum_{j' \in \mathcal{J}} \sum_{k' \in \mathcal{K}_{-k}} d_{k'j'}\, p_{k'}^{\text{dl}} \, \lvert ( \bm \gamma^{\mathrm{dl}}_{kjj'} )^T  \bm w_{k'j'}\rvert^2  \nonumber \\
		&+ \sum_{j' \in \mathcal{J}_{-j} }\sum_{k' \in \mathcal{K}_{-k}}  u_{k'j'} \, p_{k'}^{\text{ul}}\, \lvert \nu_{kk'jj'}\rvert^2 \Big),\\
		\chi_{kj}^{\text{ul}\star} &= u_{kj} \, p_{k}^{\text{ul}}  \lVert \bm \gamma^{\mathrm{ul}}_{kjj} \rVert^2 	\nonumber\\ \Big/ & \Big( \sigma^2 +  \sum_{j' \in \mathcal{J}} \sum_{k' \in \mathcal{K}_{-k}} u_{k'j'} \, p_{k'}^{\text{ul}} \, \lvert \bm v_{kj}^T \bm \gamma^{\mathrm{ul}}_{k'jj'}\rvert^2 \Big).
	\end{align}
\end{subequations}
The objective is separable in $\chi^{\text{dl}}_{kj}$ and $\chi^{\text{ul}}_{kj}$; therefore, the variables are updated simultaneously.

\subsubsection{Update {$\bm \xi^{\mathrm{dl}}$ and $\bm \xi^{\mathrm{ul}}$}}
The objective $f_2$ is strictly concave in $\xi^{\text{dl}}_{kj}$ and $\xi^{\text{ul}}_{kj}$. Hence, by keeping other parameters fixed, the unique optimal value of $\xi^{\text{dl}}_{kj}$ and $\xi^{\text{ul}}_{kj}$ can be obtained from $\partial f_2 / \partial \xi^{\text{dl}}_{kj} = 0$ and $\partial f_2 / \partial \xi^{\text{ul}}_{kj} = 0$, respectively, as
\begin{subequations}\label{eq:16}
	\begin{align}
		\xi^{\text{dl}\star}_{kj} = \,&d_{kj }\sqrt{(1+  \chi^{\text{dl}}_{kj}) p_{k}^{\text{dl}}} \, \lVert\bm \gamma^{\text{dl}}_{kjj}  \rVert  \nonumber \\
		\Big/ &\Big( \sigma^2  +\sum_{j' \in \mathcal{J}} \sum_{k' \in \mathcal{K}} d_{k'j'}\, p_{k'}^{\text{dl}} \, \lvert (\bm \gamma^{\text{dl}}_{kjj'})^T \bm w_{k'j'} \rvert^2    \nonumber \\
		& + \sum_{j' \in \mathcal{J}_{-j}} \sum_{k' \in \mathcal{K}_{-k}}  u_{k'j'} \, p_{k'}^{\text{ul}} \lvert \, \nu_{kk'jj'}\rvert^2 \Big), \\
		\xi^{\text{ul}\star}_{kj}  = \,&u_{kj }\sqrt{(1+  	\chi^{\text{ul}}_{kj}) p_{k}^{\text{ul}}} \, \lVert\bm \gamma^{\text{ul}}_{kjj} \rVert  \nonumber\\
		&\Big/\Big( \sigma^2  + \sum_{j' \in \mathcal{J}} \sum_{k' \in \mathcal{K}} u_{k'j'} \, p_{k'}^{\text{ul}}   \lvert \bm v^T_{kj} \bm \gamma^{\text{ul}}_{k'jj'} \rvert^2  \Big).
	\end{align}
\end{subequations}
The objective is separable in $\xi^{\text{dl}}_{kj}$ and $\xi^{\text{ul}}_{kj}$; therefore, the variables are updated simultaneously.

\subsubsection{Update {$\bm d,   \bm u, \bm p^{\mathrm{dl}}$, and $\bm p^{\mathrm{ul}}$}}
We jointly update $(\bm d,  \bm u, \bm p^{\text{dl}}, \bm p^{\text{ul}})$ by solving
\begin{subequations}\label{eq:17}
	\begin{align}
		& \underset{ \substack{\bm d,  \bm u, \bm p^{\text{dl}}, \bm p^{\text{ul}}}}{\text{max}}  f_2( \bm d, \bm r, \bm u, \bm p^{\text{dl}}, \bm p^{\text{ul}}, \bm \chi^{\text{dl}\star} , \bm \chi^{\text{ul}\star}, \bm \xi^{\text{dl}\star},  \bm \xi^{\text{ul}\star} ) \\
		& \qquad \text{s.t. } 
		\eqref{eq:9b}-\eqref{eq:9g}.
	\end{align}
\end{subequations}
This mathematical program involves the product of the square root of the power variable and a binary variable, as well as the product of the power variable and a binary variable. For ease of solution, we reformulate the objective $f_2$ as $f_3$ in \eqref{eq:18}, shown at the bottom of the page. The reformulation in $f_3$ reduces the number of required auxiliary variables. Specifically, instead of introducing separate auxiliary variables to model $d_{kj}\sqrt{p_k^{\mathrm{dl}}}$ and $d_{k'j'}p_{k'}^{\mathrm{dl}}$, the reformulation enables both terms to be represented using the same set of auxiliary variables, as shown next. Further, in the following proposition, we establish that they are the equivalent.
\begin{proposition}
	\label{prop:2}
	The objective functions $f_2$ and $f_3$ are equal for all $\bm d, \bm r, \bm u, \bm p^{\text{dl}}, \bm p^{\text{ul}}, \,\bm \chi^{\text{dl}} , \bm \chi^{\text{ul}}, \bm \xi^{\text{dl}}, \bm \xi^{\text{ul}}$.
\end{proposition}
\begin{IEEEproof}
		By interchanging the order of sums and relabeling the indices for the interference sums, we obtain the equality.
\end{IEEEproof}
Considering the above proposition, replacing  objective $f_2$ by $f_3$ in \eqref{eq:17} leads to an equivalent problem. Now, with the reformulated objective $f_3$, to handle both power variables and its square-root, we substitute
\begin{align}
	\setcounter{MYtempeqncnt}{\value{equation}}
	\setcounter{equation}{18}
	t_{k}^{\mathrm{dl}} = \sqrt{p_{k}^{\mathrm{dl}}}, \quad t_{k}^{\mathrm{ul}} = \sqrt{p_{k}^{\mathrm{ul}}},  \quad k = 1,\cdots,K, \label{eq:19}
\end{align}
and define $	\bm t^{\mathrm{dl}} = [t_{k}^{\mathrm{dl}}]_{k=1}^K, \bm t^{\mathrm{ul}} = [t_{k}^{\mathrm{ul}}]_{k=1}^K$. This substitution is one-to-one for $t_{k}^{\mathrm{dl}} \geq 0$ and $t_{k}^{\mathrm{ul}} \geq 0$. Therefore, both  problems are equivalent. Then, the continuous-binary product can be handled by the following substitution
\begin{align}
	z^{\mathrm{dl}}_{kj} = d_{kj} t_k^{\mathrm{dl}}, \quad z^{\mathrm{ul}}_{kj} = u_{kj} t_k^{\mathrm{ul}}, \quad \text{for all $k$ and $j$} \label{eq:20}
\end{align}
where the auxiliary variables are $\bm z^{\mathrm{dl}} = [z^{\mathrm{dl}}_{kj}]_{k=1,j=1}^{k=K, j=J}$ and $\bm z^{\mathrm{ul}} = [z^{\mathrm{ul}}_{kj}]_{k=1,j=1}^{k=K, j=J}$. We use standard big-$M$~\cite{Dejans2025} constraints to enforce \eqref{eq:20}. In particular, consider a continuous variable $p \in \mathbb{R}$, with $0\leq p \leq M$, and a binary  variable $b \in \{0,1\}$. Then, $z = bp$ is equivalent to the following inequalities~\cite{Glover1975, Adams2005}:
\begin{align}
	 \text{(i)}\,  z \geq 0, & \quad \text{(ii)}\,   z \leq Mb, \label{eq:21}\\
  \text{(iii)}\, z \leq p, & \quad \text{(iv)}\, z \geq p - M(1 - b). \label{eq:22}
\end{align}
Thus, we can implement \eqref{eq:20} as 
\begin{subequations}\label{eq:23}
	\begin{align}
		&0 \leq z_{kj}^{\text{dl}} \leq  t^{\text{dl}}_{k}, 0 \leq z_{kj}^{\text{ul}} \leq  t^{\text{ul}}_{k}, \label{eq:23a}\\
		&t_k^{\text{dl}} - M(1-d_{kj}) \leq z_{kj}^{\text{dl}} \leq Md_{kj}, \label{eq:23b}\\
		&t_k^{\text{ul}} - M(1-u_{kj}) \leq z_{kj}^{\text{ul}} \leq Mu_{kj}.\label{eq:23c}
	\end{align}
\end{subequations}
For equivalence to hold, we need to select a proper value for $M$, and this will be described below in section~\ref{sec:selectM}. Further, observe that constraint \eqref{eq:9a} is not amenable for most solvers. An equivalent, and more solver-friendly, version is
\begin{align}
	&\Big\lvert \sum_{j \in \mathcal{J}} d_{kj} r_{j} - \sum_{j \in \mathcal{J}}  u_{kj} r_{j} \Big \rvert  \nonumber \\
	&\quad \qquad \qquad  \leq M( 2- \sum_{j \in \mathcal{J}} d_{kj} - \sum_{j \in \mathcal{J}} u_{kj} ), \forall k. \label{eq:24} 
\end{align}
The left-hand side of the above constraint measures the difference in spin (band assignment) between the satellite serving the \cgls{dl} and the one serving the \cgls{ul}. The right-hand side, using a large number $M$, acts as a conditional switch: if both \cgls{dl} and \cgls{ul} links are assigned, the right side is zero, forcing the spins to match and preventing simultaneous transmission and reception on same bands; if at most one link is assigned, the right side is large and the constraint is inactive.
\begin{figure*}[b]
	\normalsize
	\setcounter{MYtempeqncnt}{\value{equation}}
	\setcounter{equation}{25}
	\hrulefill 
	\vspace{0.5em}
	\begin{minipage}{\textwidth}
		\begin{align}
			f_4(\bm d, \bm r, \bm u, &\bm t^{\text{dl}}, \bm t^{\text{ul}}, \bm z^{\mathrm{dl}}, \bm z^{\mathrm{ul}},\bm \chi^{\text{dl}} , \bm \chi^{\text{ul}}, \bm \xi^{\text{dl}}, \bm \xi^{\text{ul}} ) =   \sum_{k \in \mathcal{K}} \sum_{j\in \mathcal{J}}  \Big[\log_2 (1+  	\chi^{\text{dl}}_{kj} )	+ \log_2 ( 1 + \chi^{\text{ul}}_{kj} ) - \chi^{\text{dl}}_{kj}  - \chi^{\text{ul}}_{kj} - (\xi_{kj}^{\text{dl}} \sigma)^2 - (\xi_{kj}^{\text{ul}} \sigma)^2
			\nonumber \\
			&+  2 \xi_{kj}^{\text{dl}} \, z_{kj}^{\mathrm{dl}}\sqrt{(1+  	\chi^{\text{dl}}_{kj} )} \, \lVert\bm \gamma^{\text{dl}}_{kjj}\rVert + 2 \xi_{kj}^{\text{ul}} \, z_{kj}^{\mathrm{ul}}\sqrt{(1+  	\chi^{\text{ul}}_{kj} )} \, \lVert\bm \gamma^{\text{ul}}_{kjj}  \rVert - (z_{kj}^{\mathrm{dl}})^2 \sum_{j' \in \mathcal{J}} \sum_{k' \in \mathcal{K}} (\xi_{k'j'}^{\text{dl}})^2  \, \lvert (\bm \gamma^{\text{dl}}_{k'j'j})^T \bm w_{kj} \rvert^2    \nonumber \\ 
			& - (z_{kj}^{\mathrm{ul}})^2 \sum_{j' \in \mathcal{J}_{-j}} \sum_{k' \in \mathcal{K}_{-k}} (\xi_{k'j'}^{\text{dl}})^2  \,\lvert \nu_{k'kj'j} \rvert^2- (z_{kj}^{\mathrm{ul}})^2  \sum_{j' \in \mathcal{J}} \sum_{k' \in \mathcal{K}} (\xi_{k'j'}^{\text{ul}})^2   \lvert \bm v^T_{k'j'} \bm \gamma^{\text{ul}}_{kj'j} \rvert^2   \Big]. \label{eq:26}
		\end{align}
	\end{minipage}
	\setcounter{equation}{\value{MYtempeqncnt}}
\end{figure*}

Finally, the problem in \eqref{eq:14} is written with big-$M$ constraints as
\begin{subequations}\label{eq:25}
	\begin{align}
		&\underset{ \substack{\bm d,  \bm u, \\\bm t^{\text{dl}}, \bm t^{\text{ul}}, \\ \bm z^{\mathrm{dl}}, \bm z^{\mathrm{ul}}}}{\text{max}}  f_4( \bm d, \bm r, \bm u, \bm t^{\text{dl}}, \bm t^{\text{ul}}, \bm z^{\mathrm{dl}},  \bm z^{\mathrm{ul}}, \bm \chi^{\text{dl}\star}, \bm \chi^{\text{ul}\star}, \bm \xi^{\text{dl}\star},   \bm \xi^{\text{ul}\star}) \\
		& \quad  \text{s.t. } \quad \eqref{eq:24}, \label{eq:25b}\\
		& \quad \qquad \sum_{k \in \mathcal{K}} (z_{kj}^{\text{dl}})^2\leq p_{j}^{\text{max}}, \text{ for all } j, \label{eq:25c}\\
		& \quad \qquad t_{k}^{\mathrm{dl}} \geq 0, t_{k}^{\mathrm{ul}} \geq 0, \quad \text{ for all } k, \label{eq:25d}\\
		& \quad  \qquad \eqref{eq:9c}, \eqref{eq:9e}-\eqref{eq:9g}, \eqref{eq:23a}-\eqref{eq:23c}. \label{eq:25e}
	\end{align}
\end{subequations}
 The objective $f_4$ is shown at the bottom in \eqref{eq:26}, which is concave in $z_{kj}^{\mathrm{dl}}$ and $z_{kj}^{\mathrm{ul}}$, and can be solved to optimality using industry grade \cgls{minlp} solvers. We use MOSEK here.

\subsubsection{Selection of $M$}
\label{sec:selectM}

For the linearization in \eqref{eq:23}, the constant $M$ must serve as an upper bound on the values of the continuous variables involved in each binary-continuous product. Further, $M$ should be selected to act as a conditional switch in \eqref{eq:24}. 
\begin{proposition}
	\label{prop:bigM}
	A common constant $M$ that is suitable for \eqref{eq:23} and \eqref{eq:24} is
	\begin{align*}
		M = \Big\lceil \max_{k,j} \big\{1,\ \sqrt{p_j^{\text{max}}},\ \sqrt{p_k^{\text{max}} }\Big\} \Big\rceil,
	\end{align*}
	where $\lceil \cdot \rceil$ denotes the ceiling operator.
\end{proposition}
\begin{IEEEproof}
	See Appendix~\ref{appendix:bigM}.
\end{IEEEproof}

With the $M$ value selected, we state the following proposition to establish the equivalence between the mathematical programs \eqref{eq:17} and \eqref{eq:25} below.
\begin{proposition}
	\label{prop:4}
	Consider the mathematical programs in \eqref{eq:17} and \eqref{eq:25}. Then, if  $(\bm d^\star,  \bm u^\star, \bm t^{\text{dl}\star}, \bm t^{\text{ul}\star},  \bm z^{\mathrm{dl}\star}, \bm z^{\mathrm{ul}\star})$ is a solution to \eqref{eq:25}, it follows that $(\bm d^\star,  \bm u^\star, \bm p^{\text{dl}\star}, \bm p^{\text{ul}\star})$ is a solution to \eqref{eq:17}, with $\bm p^{\text{dl}\star} = (\bm t^{\text{dl}\star})^2$ and $\bm p^{\text{ul}\star} = (\bm t^{\text{ul}\star})^2$.
\end{proposition}
\begin{IEEEproof}
	See Appendix \ref{appendix:4}.
\end{IEEEproof}
With above equivalence, we now state the full algorithm to solve the original problem in \eqref{eq:9} in the next subsection.

\subsection{Algorithm}

The complete solution procedure is summarized in Algorithm~\ref{algo:1}, where the iteration index of the alternating optimization procedure is denoted by $i$. The algorithm enumerates all feasible binary scheduling vectors $\bm r \in \{0,1\}^J$. For each fixed realization of $\bm r$, the continuous optimization variables $\bm d_{(0)}$, $\bm u_{(0)}$, $\bm p^{\mathrm{dl}}_{(0)}$, and $\bm p^{\mathrm{ul}}_{(0)}$ are initialized to a feasible point that satisfies all constraints. The scheduling variables $\bm d_{(0)}$ and $\bm u_{(0)}$ are initialized such that, for each \cgls{ue}~$k$, we randomly choose a satellite $j \in \mathcal{J}$ and set $d_{kj} = u_{kj} = 1$, which satisfies \eqref{eq:9c}. Further, this initialization also satisfies the constraint~\eqref{eq:24}. Then, we allocate equal power for the \cgls{dl} to the scheduled \cglspl{ue}, i.e., $p_{k}^{\mathrm{dl}} = {p_j^{\mathrm{max}}}/{\sum_{k \in \mathcal{K}} d_{kj}}$. The \cgls{ul} power at \cgls{ue}~$k$ is initialized as $p_{k}^{\mathrm{ul}} = p_{k}^{\mathrm{max}}$. For a given $\bm r$, the resulting problem is then solved iteratively via an alternating (block-coordinate) optimization procedure. At iteration $i$, the auxiliary variables
$\bm \chi^{\mathrm{dl}}_{(i)}$ and $\bm \chi^{\mathrm{ul}}_{(i)}$ are first updated according to~\eqref{eq:15}. Subsequently, the variables
$\bm \xi^{\mathrm{dl}}_{(i)}$ and $\bm \xi^{\mathrm{ul}}_{(i)}$ are updated using~\eqref{eq:16}. With these auxiliary variables fixed, the remaining variables
$\big(\bm d_{(i)}, \bm u_{(i)}, \bm p^{\mathrm{dl}}_{(i)}, \bm p^{\mathrm{ul}}_{(i)}\big)$ are jointly updated by solving a \cgls{minlp} problem in \eqref{eq:25} using MOSEK. These steps are repeated until convergence is reached. Here, convergence implies the objective improvement is below a predefined threshold.  After convergence for all feasible realizations of $\bm r$, the algorithm selects the solution $\big(\bm d^\star, \bm r^\star, \bm u^\star,
\bm p^{\mathrm{dl}\star}, \bm p^{\mathrm{ul}\star},
\bm \chi^{\mathrm{dl}\star}, \bm \chi^{\mathrm{ul}\star},
\bm \xi^{\mathrm{dl}\star}, \bm \xi^{\mathrm{ul}\star}\big)$
that yields the maximum objective value.
Finally, the optimal objective value $f_0$ is computed from~\eqref{eq:10}
using $(\bm d^\star, \bm r^\star, \bm u^\star,
\bm p^{\mathrm{dl}\star}, \bm p^{\mathrm{ul}\star})$. The convergence properties of the proposed algorithm are discussed next.

\begin{algorithm}[t]
	\caption{Solution for optimization \eqref{eq:9}}
	\begin{algorithmic}[1]
		\ForAll{$\bm r = [r_1, \cdots, r_J] \in \{0,1\}^J$}
		\State \textbf{Initialize:} feasible
		$\bm d_{(0)}, \bm u_{(0)},
		\bm p^{\mathrm{dl}}_{(0)}, \bm p^{\mathrm{ul}}_{(0)}$
		\Repeat
				\State $i \leftarrow i+1$ 
		\State update $\bm \chi^{\mathrm{dl}}_{(i)}, \bm \chi^{\mathrm{ul}}_{(i)}$ by \eqref{eq:15}
		\State update $\bm \xi^{\mathrm{dl}}_{(i)}, \bm \xi^{\mathrm{ul}}_{(i)}$ by \eqref{eq:16}
		\State update $\bm d_{(i)},  \bm u_{(i)},
		\bm p^{\mathrm{dl}}_{(i)}, \bm p^{\mathrm{ul}}_{(i)}$
		by solving \eqref{eq:25}
		\Until{objective improvement is below a threshold}
		\EndFor
		\State select
		$\big(\bm d^\star, \bm r^\star, \bm u^\star,
		\bm p^{\mathrm{dl}\star}, \bm p^{\mathrm{ul}\star},
		\bm \chi^{\mathrm{dl}\star}, \bm \chi^{\mathrm{ul}\star},
		\bm \xi^{\mathrm{dl}\star}, \bm \xi^{\mathrm{ul}\star}\big)$
		\State compute the optimal objective $f_0$ from \eqref{eq:10}
		using $(\bm d^\star, \bm r^\star, \bm u^\star,
		\bm p^{\mathrm{dl}\star}, \bm p^{\mathrm{ul}\star})$
	\end{algorithmic}
	\label{algo:1}
\end{algorithm}

\begin{theorem}
\label{prop:5}
    Let $\big\{\big(\bm d_{(i)},  \bm u_{(i)}, \bm p^{\mathrm{dl}}_{(i)}, \bm p^{\mathrm{ul}}_{(i)}, \bm \chi^{\mathrm{dl}}_{(i)}, \bm \chi^{\mathrm{ul}}_{(i)}, \bm \xi^{\mathrm{dl}}_{(i)}, \bm \xi^{\mathrm{ul}}_{(i)}\big)\big\}$ be the sequence generated by the Algorithm~\ref{algo:1}. Then the following holds:
    \begin{enumerate}
        \item The sequence has at least one limit point, denoted by
        \begin{align*}
            \big(\bm d^{\star}, \bm u^{\star}, \bm p^{\mathrm{dl}{\star}}, \bm p^{\mathrm{ul}{\star}}, \bm \chi^{\mathrm{dl}{\star}}, \bm \chi^{\mathrm{ul}{\star}}, \bm \xi^{\mathrm{dl}{\star}}, \bm \xi^{\mathrm{ul}{\star}}\big).
        \end{align*}
        \item The continuous variables
        $(\bm p^{\mathrm{dl}\star}, \allowbreak \bm p^{\mathrm{ul}\star}, \allowbreak \bm \chi^{\mathrm{dl}\star},\allowbreak \bm \chi^{\mathrm{ul}\star},\allowbreak \bm \xi^{\mathrm{dl}\star}, \allowbreak\bm \xi^{\mathrm{ul}\star})$ are a stationary point of $f_2$, where  $(\bm d^\star, \bm u^\star)$ maximizes $f_2$.
        
   		 \item Since the algorithm performs exhaustive search over $\bm r$, the obtained solution is globally optimal \cgls{wrt} $\bm r$.
\end{enumerate}
\end{theorem}

\begin{IEEEproof}
	See Appendix~\ref{appendix:5}.
\end{IEEEproof}

With convergence of the proposed algorithm, we now verify the key implications of dynamic \cgls{fdd} by numerical simulations in the next section.

\section{Numerical Results}
\label{sec:4}
	
In this section, we evaluate the proposed dynamic \cgls{fdd} approach against conventional static \cgls{fdd} systems. We consider a scenario with $J$ satellites over a reference location. The satellites are randomly positioned at elevation angles ranging from $30^\circ$ to $80^\circ$ from the reference location and according to the minimum elevation criteria from 3GPP, thereby capturing a wide range of geometric configurations representative of existing satellite constellations. The serving area of the \cglspl{ue} is assumed to be a rural environment, modeled as a circular region with a radius of 10\,km, roughly an area of 314\,km$^2$, where terrestrial communication links are absent. A total of $K$ \cglspl{ue} are randomly distributed within this area and are served by the $J$ satellites operating at an altitude of 500\,km. Each satellite is equipped with $N = 8 \times 8$ antennas arranged in a \gls{upa} configuration. The inter-\cgls{ue} channels are modeled as purely \cgls{los}. The maximum transmit power at each satellite is set to $20$\,W, i.e., $p_j^{\mathrm{max}} = 20, \forall j$, while each \cgls{ue} transmits with a maximum power of $2$\,W in the \cgls{ul} direction, i.e., $p_k^{\mathrm{max}} = 2, \forall k$~\cite{3gpp}. We consider two frequency bands $\Omega_2 = 1.6$\,GHz and  $\Omega_1 = 2.4$\,GHz. Furthermore, equal bandwidths of $B_1 = B_2 = 10$\,MHz are assumed for both frequency bands. Corresponding to these system parameters, the big-$M$ parameter in Algorithm~\ref{algo:1} is $M = 5$. In the results, we consider two cases as the baseline results, first is where the spin of each satellite is set to 0, termed as ``All spin 0", and another is where the spin of each satellite is set to 1, termed as ``All spin 1", these cases denotes the conventional \cgls{fdd}, where we fix beforehand the bands for \cgls{ul} and \cgls{dl}. For these cases, the optimization problem in \eqref{eq:9} is solved using Algorithm~\ref{algo:1} but with fixed $\bm r$ in step-1. These baselines are then compared with the proposed ``dynamic \cgls{fdd}" case, where we select the best objective value over all the possible configurations of the spins, also termed as ``optimized spin".

\begin{figure}[t]
	\centering
	{\footnotesize
	\begin{tikzpicture}
		\begin{axis}[
			width=1\linewidth,
			height=0.7\linewidth,
			mygrid,
			xlabel={Sum-rate (bits/sec/Hz)},
			ylabel={CDF},
			ymin=0,
			ymax=1,
			grid style={dotted, very thick},
			legend style={
				at={(0.03,0.97)},
				anchor=north west,
				draw=gray,
				rounded corners=2pt,
				fill=white,
			},
			legend image post style={sharp corners},
			every axis plot/.append style={thick}
			]
			\addplot+[plotspinzero, mark=triangle*, mark repeat=20, mark options = {xshift=-1, scale = 1.6}] table[x=x, y=cdf, col sep=comma] {generated/cdf_spin0.csv};
			\addlegendentry{All spin 0}
			
			\addplot+[plotspinone, mark repeat=25, mark options = {xshift=1, scale = 1.6}] table[x=x, y=cdf, col sep=comma] {generated/cdf_spin1.csv};
			\addlegendentry{All spin 1}
			
			\addplot+[plotopt, mark=none] table[x=x, y=cdf, col sep=comma] {generated/cdf_optimized.csv};
			\addlegendentry{\shortstack[l]{Optimized spin\\(dynamic FDD)}}
		\end{axis}
	\end{tikzpicture}}
	\caption{CDF of the objective value $f_0$ for \(J=\CdfJ\) satellites and \(K=\CdfK\) users.}
	\label{fig:5}
\end{figure}
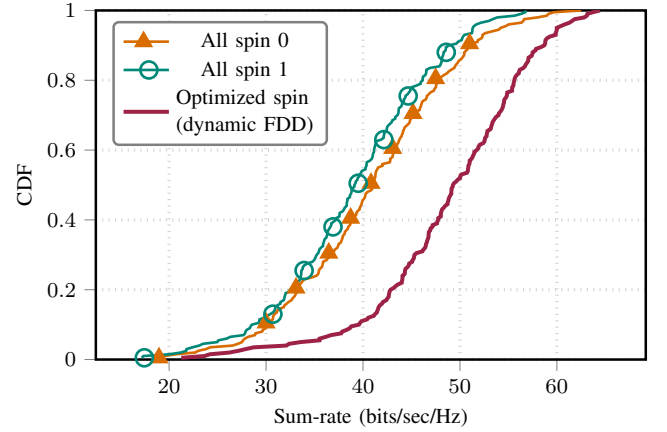

\begin{figure}[t]
	\centering{\footnotesize
	\begin{tikzpicture}
		\begin{axis}[
			width=1\linewidth,
			height=0.7\linewidth,
			mygrid,
			xlabel={Number of users \(K\)},
			ylabel={Average sum-rate (bits/sec/Hz)},
			xtick=data,
		grid style={dotted, very thick}, 
			legend style={
				at={(0.03,0.97)},
				anchor=north west,
				draw=gray,
				rounded corners=2pt,
				fill=white,
			},
			legend image post style={sharp corners},
			every axis plot/.append style={very thick}
			]
			\addplot+[plotspinzero, mark options = {xshift=-2, scale = 1.6}] table[x=K, y=spin0, col sep=comma] {generated/avg_vs_k.csv};
			\addlegendentry{All spin 0}
			
			\addplot+[plotspinone,  mark options = {xshift=2, scale = 1.6}] table[x=K, y=spin1, col sep=comma] {generated/avg_vs_k.csv};
			\addlegendentry{All spin 1}
			
			\addplot+[plotopt] table[x=K, y=optimized, col sep=comma] {generated/avg_vs_k.csv};
			\addlegendentry{\shortstack[l]{Optimized spin\\(dynamic FDD)}}
		\end{axis}
	\end{tikzpicture}}
	\caption{Average sum-rate $f_0$ versus number of users \(K\) for fixed \(J=\FixedJ\) satellites.}
	\label{fig:6}
\end{figure}
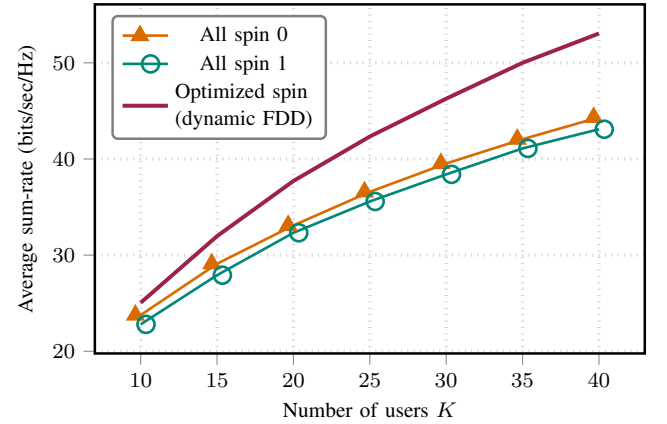

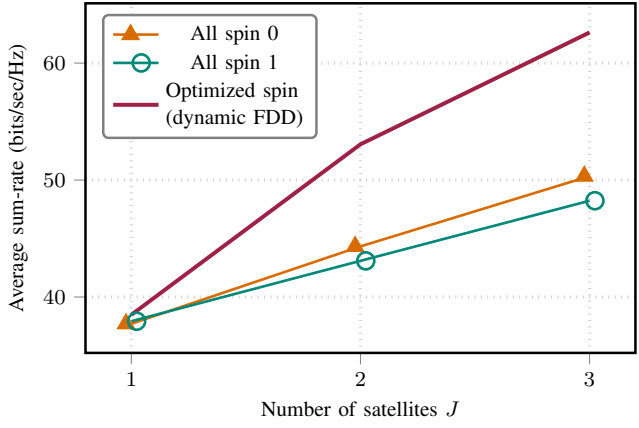
\begin{figure}[t]
	\centering{\footnotesize
	\begin{tikzpicture}
		\begin{axis}[
			width=1\linewidth,
			height=0.7\linewidth,
			mygrid,
			xlabel={Number of satellites \(J\)},
			ylabel={Average sum-rate (bits/sec/Hz)},
			xtick=data,
			grid style={dotted, very thick},
			legend style={
				at={(0.03,0.97)},
				anchor=north west,
				draw=gray,
				rounded corners=2pt,
				fill=white,
			},
			legend image post style={sharp corners},
			every axis plot/.append style={very thick}
			]
			\addplot+[plotspinzero,  mark options = {xshift=-2, scale = 1.6}] table[x=J, y=spin0, col sep=comma] {generated/avg_vs_j.csv};
			\addlegendentry{All spin 0}
			
			\addplot+[plotspinone,  mark options = {xshift=2, scale = 1.6}] table[x=J, y=spin1, col sep=comma] {generated/avg_vs_j.csv};
			\addlegendentry{All spin 1}
			
			\addplot+[plotopt] table[x=J, y=optimized, col sep=comma] {generated/avg_vs_j.csv};
			\addlegendentry{\shortstack[l]{Optimized spin\\(dynamic FDD)}}
		\end{axis}
	\end{tikzpicture}}
	\caption{Average sum-rate $f_0$ versus \(J\) for fixed number of users $(K=\FixedK)$.}
		\label{fig:7}
\end{figure}

\begin{figure}[t]
	\centering{\footnotesize
	\begin{tikzpicture}
		\begin{axis}[
			width=1\linewidth,
			height=0.7\linewidth,
			mygrid,
			xlabel={Number of users \(K\)},
			ylabel={Relative improvement (\%)},
			xtick=data,
			grid style={dotted, very thick},
			legend style={
				at={(0.03,0.97)},
				anchor=north west,
				draw=gray,
				rounded corners=2pt,
				fill=white,
			},
			legend image post style={sharp corners},
			every axis plot/.append style={very thick}
			]
			\addplot+[solid, color = mypantone, mark options = {fill=mypantone, scale=1.5}] table[x=K, y=vs_spin0, col sep=comma] {generated/improvement_vs_k.csv};
			\addlegendentry{w.r.t. all spin 0}
			
			\addplot+[solid, color = cyan!80!black, mark options = {fill=cyan!80!black, scale=1.5}] table[x=K, y=vs_spin1, col sep=comma] {generated/improvement_vs_k.csv};
			\addlegendentry{w.r.t. all spin 1}
		\end{axis}
	\end{tikzpicture}}
	\caption{Relative improvement of the optimized spin objective versus number of users \(K\) for fixed \(J=\FixedJ\) satellites, measured using average objective values.}
		\label{fig:8}
\end{figure}
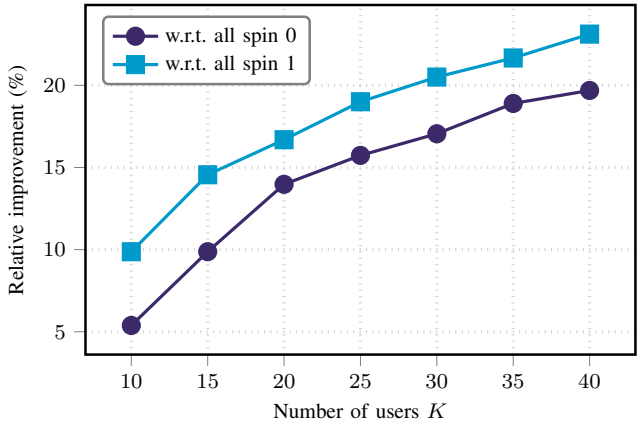

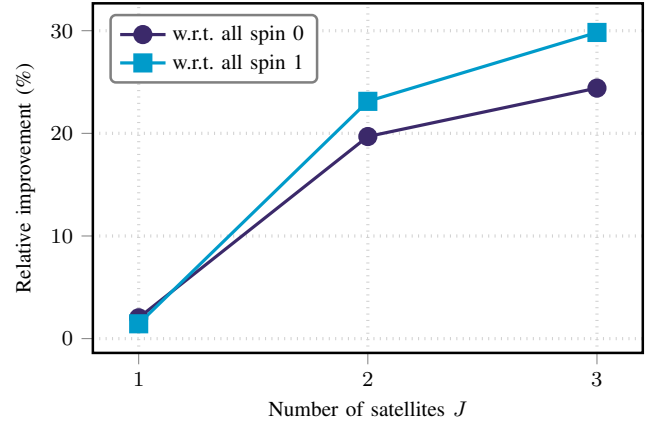
\begin{figure}[t]
	\centering{\footnotesize
	\begin{tikzpicture}
		\begin{axis}[
			width=1\linewidth,
			height=0.7\linewidth,
			mygrid,
			xlabel={Number of satellites \(J\)},
			ylabel={Relative improvement (\%)},
			xtick=data,
			grid style={dotted, very thick},
			legend style={
				at={(0.03,0.97)},
				anchor=north west,
				draw=gray,
				rounded corners=2pt,
				fill=white,
			},
			legend image post style={sharp corners},
			every axis plot/.append style={very thick}
			]
			\addplot+[solid, color = mypantone, mark options = {fill=mypantone, scale=1.5}] table[x=J, y=vs_spin0, col sep=comma] {generated/improvement_vs_j.csv};
			\addlegendentry{w.r.t. all spin 0}
			
			\addplot+[solid, color = cyan!80!black, mark options = {fill=cyan!80!black,  scale=1.5}] table[x=J, y=vs_spin1, col sep=comma] {generated/improvement_vs_j.csv};
			\addlegendentry{w.r.t. all spin 1}
		\end{axis}
	\end{tikzpicture}}
	\caption{Relative improvement of the optimized spin objective versus \(J\) satellites for fixed \(K=\FixedK\) users, measured using average objective values.}
		\label{fig:9}
\end{figure}

Fig.~\ref{fig:5} shows the \gls{cdf} of the objective value $f_0$ for $J=3$ satellites and $K=25$ users. It can be observed that the proposed optimized spin (dynamic \cgls{fdd}) configuration consistently outperforms both fixed-spin baselines, which indicates a distribution-wide performance improvement. More precisely, at the $10^{\text{th}}$ percentile, the dynamic \cgls{fdd} achieves $39.59$ bits/s/Hz, compared with $30.19$ bits/s/Hz for the all spin 0 case and $28.54$ bits/s/Hz for the all spin 1 case, corresponding to gains of $31.1\%$ and $38.7\%$, respectively. At the median, dynamic \cgls{fdd}  provides gains of $20.2\%$ over all spin 0 and $26.0\%$ over all spin 1, while at the $90^{\text{th}}$ percentile the gains remain significant at $14.8\%$ and $20.0\%$, respectively. These results show that the proposed dynamic \gls{fdd} scheme improves not only the high-interference regime of the performance distribution, but also the medium and low-interference regimes.
	
Fig.~\ref{fig:6} presents the average sum-rate $f_0$ (bits/s/Hz) as a function of the number of users $K$ for a fixed number of satellites $J=2$. We can observe that the average sum-rate increases with $K$ for all considered schemes due to the larger scheduling diversity. However, the dynamic \cgls{fdd} consistently achieves the highest performance, and its advantage becomes more pronounced as the system load increases, implying higher efficiency in high-interference scenarios.  In particular, at $K=10$, the dynamic \cgls{fdd} attains $25.04$ bits/s/Hz, which is $5.38\%$ higher than the all spin 0 baseline and $9.87\%$ higher than the all spin 1 baseline. When $K$ increases to $40$, the average sum-rate rises to $53.06$ bits/s/Hz, and the corresponding gains increase to $19.69\%$ and $23.12\%$, respectively. This trend indicates that the proposed dynamic \cgls{fdd} becomes increasingly beneficial in high-interference scenarios, where interference-aware band-direction assignment can be exploited more effectively.
	
Fig.~\ref{fig:7} illustrates the average sum-rate versus the number of satellites $J$ for a fixed number of users, i.e., $K=40$. Increasing $J$ improves the performance of all schemes, since additional satellites provide more spatial resources and a larger aggregate transmission capability. Nevertheless,  dynamic \cgls{fdd}  scales more favorably than the baselines. For $J=1$, the gain is naturally small, namely $2.03\%$ over all spin 0 and $1.45\%$ over all spin 1, because inter-satellite interference is essentially absent in the single-satellite case. In contrast, when $J=2$, the optimized scheme reaches $53.06$ bits/s/Hz, corresponding to gains of $19.69\%$ and $23.12\%$, respectively, and for $J=3$ it further increases to $62.62$ bits/s/Hz, yielding gains of $24.41\%$ and $29.83\%$, respectively. These results clearly show that dynamic \gls{fdd} is particularly effective in dense multi-satellite deployments, where coordinated interference management becomes crucial.
	
Fig.~\ref{fig:8} quantifies the relative improvement achieved by the dynamic \cgls{fdd} as a function of the number of users $K$ for fixed $J=2$. The relative improvement is defined as the difference between the objective value achieved by dynamic \cgls{fdd} and that of the baseline schemes (all spin 0 or all spin 1), normalized by the corresponding baseline value. It can be seen that the gain increases monotonically with the number of users, which shows that the benefit of the proposed method becomes more pronounced as the number of users increases. Specifically, the improvement over the all spin 0 baseline rises from $5.38\%$ at $K=10$ to $19.69\%$ at $K=40$, while the improvement over the all spin 1 baseline increases from $9.87\%$ to $23.12\%$. This consistent widening of the performance gap indicates that fixed spin assignments are less capable of coping with the more complex interference patterns induced by larger users and satellites. By contrast, the optimized dynamic \gls{fdd} can exploit the additional flexibility more efficiently, thereby delivering increasingly larger gains under higher traffic load.
	
Fig.~\ref{fig:9} reports the relative improvement of the dynamic \cgls{fdd} objective as a function of the number of satellites $J$ for fixed $K=40$. The improvement is modest in the single-satellite case, with gains of only $2.03\%$ over all spin 0 and $1.45\%$ over all spin 1, which is expected because the \cgls{dof} that motivates dynamic \gls{fdd} is minimal when $J=1$. Once multiple satellites are active, however, the gain increases sharply. At $J=2$, the relative improvement reaches $19.69\%$ and $23.12\%$, and at $J=3$ it further rises to $24.41\%$ and $29.83\%$ with respect to the two baselines. Hence, this figure provides direct evidence that the proposed dynamic \cgls{fdd} is most valuable in the interference-limited regime for which it is designed, and that its advantage strengthens as the satellite deployment becomes denser.

\section{Conclusions}
\label{sec:5}

We investigated a new framework for spectrum sharing in dense \cgls{leo} satellite mega-constellations envisioned for future 6G networks. Motivated by the growing interference caused by satellite deployments operating over shared frequency bands, we revisited the limitations of conventional \cgls{fdd} systems with fixed \cgls{dl} and \cgls{ul} band allocations. To address this challenge, we proposed a dynamic  \cgls{fdd}  band framework that adapts the \cgls{dl} and \cgls{ul} bands based on the system considerations, which bring a new \cgls{dof} for resource allocation. Utilizing this new \cgls{dof}, we formulated a joint user scheduling, band assignment, and power allocation.  Efficient solutions were obtained through a combination of equivalence transformations, alternating optimization, and mixed-integer solvers. Numerical results confirm that dynamic \cgls{fdd} provides substantial performance gains compared to conventional \cgls{fdd} operation, achieving up to a 30\% increase in throughput in dense satellite deployments. Additionally, the performance improves as the system density, i.e., number of \cglspl{ue} and satellites increase.  These findings highlight the potential of flexible band assignment as a key enabler for interference management in satellite mega-constellations. The proposed approach offers a practical and effective solution for improving spectral efficiency in next-generation non-terrestrial networks.

\appendices

\section{Proof of Proposition~\ref{prop:1}}
\label{appendix:1}

Observe that the function $f_1$ is smooth \cgls{wrt} $\chi_{kj}^{\mathrm{dl}}$ and $\chi_{kj}^{\mathrm{ul}}$, for $\chi_{kj}^{\mathrm{dl}}, \chi_{kj}^{\mathrm{ul}} \in \mathbb{R}^+$. Further, the two variables $\chi_{kj}^{\mathrm{dl}}$ and $\chi_{kj}^{\mathrm{ul}}$ are separable keeping other variables fixed, i.e., cross diagonal elements in the Hessian matrix are zero. Its second derivatives satisfy ${\partial^2 f_1}/{(\partial \chi^{\mathrm{ul}}_{kj})^2} < 0$ and ${\partial^2 f_1}/{(\partial \chi^{\text{dl}}_{kj})^2} < 0$ within the domain. Hence, $f_1$ is strictly concave in $\chi_{kj}^{\mathrm{dl}}$ and $\chi_{kj}^{\mathrm{ul}}$. Similarly, $f_2$ is smooth and strictly concave \cgls{wrt} $\xi_{kj}^{\mathrm{dl}}$ and $\xi_{kj}^{\mathrm{ul}}$.

Consider the program in $\eqref{eq:14}$ its feasible set is $(\bm d,\allowbreak \bm r,\allowbreak \bm u, \allowbreak \bm p^{\mathrm{dl}}, \allowbreak \bm p^{\mathrm{ul}},\allowbreak \bm \chi^{\mathrm{dl}} , \bm \chi^{\mathrm{ul}}, \bm \xi^{\mathrm{dl}}, \bm \xi^{\mathrm{ul}}) $, where $(\bm d, \bm r, \bm u, \bm p^{\mathrm{dl}}, \bm p^{\mathrm{ul}} ) \in \mathcal{F}$, $\chi_{kj}^{\mathrm{dl}} \in \mathbb{R}^+, \chi_{kj}^{\mathrm{ul}} \in \mathbb{R}^+, \xi_{kj}^{\mathrm{dl}} \in \mathbb{R}, \xi_{kj}^{\mathrm{ul}} \in \mathbb{R}$, with $\mathcal{F} = \{(\bm d, \bm r, \bm u, \bm p^{\mathrm{dl}}, \bm p^{\mathrm{ul}}): \, \textit{satisfying }\eqref{eq:9b}-\eqref{eq:9g} \}$. For fixed $(\bm d, \bm r, \bm u, \bm p^{\mathrm{dl}}, \bm p^{\mathrm{ul}}, \bm \chi^{\mathrm{dl}} , \bm \chi^{\mathrm{ul}})$, we can obtain a unique optimal value as $\xi_{kj}^{\mathrm{dl}\star}$ and $\xi_{kj}^{\mathrm{ul}\star}$ that maximizes the objective $f_2$, due to strict concavity.  These unique values are  shown in \eqref{eq:15}. 

Now, if we put \eqref{eq:15} in $f_2$, we can exactly recover $f_1$. This implies, if we obtain a solution to $\eqref{eq:14}$ as $(
\bm d^\star, \allowbreak\bm r^\star, \allowbreak
\bm u^\star, \allowbreak
\bm p^{\mathrm{dl}\star}, \allowbreak \bm p^{\mathrm{ul}\star},\allowbreak
\bm \chi^{\mathrm{dl}\star}, \allowbreak \bm \chi^{\mathrm{ul}\star},\allowbreak \bm \xi^{\mathrm{dl}\star}, \allowbreak \bm \xi^{\mathrm{ul}\star} 
)$, then $(
\bm d^\star, \allowbreak \bm r^\star, \allowbreak
\bm u^\star, \allowbreak
\bm p^{\mathrm{dl}\star}, \allowbreak\bm p^{\mathrm{ul}\star}, \allowbreak
\bm \chi^{\mathrm{dl}\star},\allowbreak \bm \chi^{\mathrm{ul}\star}\allowbreak
)$, is a solution to $\eqref{eq:10}$. 

Similarly, we can also prove that if  $(
\bm d^\dagger,\allowbreak \bm r^\dagger, \allowbreak
\bm u^\dagger, \allowbreak
\bm p^{\mathrm{dl}\dagger},\allowbreak \bm p^{\mathrm{ul}\dagger}, \allowbreak
\bm \chi^{\mathrm{dl}\dagger},\allowbreak \bm \chi^{\mathrm{ul}\dagger}\allowbreak
)$ is a solution to $\eqref{eq:13}$, then  $(
\bm d^\dagger,\allowbreak\bm r^\dagger, \allowbreak
\bm u^\dagger, \allowbreak
\bm p^{\mathrm{dl}\dagger}, \allowbreak\bm p^{\mathrm{ul}\dagger}\allowbreak
)$ is a solution to $\eqref{eq:9}$.	
\IEEEQED

\section{Proof  of Proposition~\ref{prop:bigM}}
	\label{appendix:bigM}
	
From \eqref{eq:19}, and as $p_k^{\mathrm{dl}} \leq p_j^{\mathrm{max}}, p_k^{\mathrm{ul}} \leq p_k^{\mathrm{max}}$, we get 
\begin{align*}
	M \geq \max_{k,j} \Big\{ \sqrt{p_j^{\text{max}}}, \sqrt{p_k^{\text{max}}} \Big\}. 
\end{align*}
The same constant $M$ is also used as a conditional switch in constraint \eqref{eq:24}. For this constraint to remain inactive when required, $M$ must satisfy
	\begin{align*}
		M \geq \max_k \, \Big\lvert \sum_{j \in \mathcal{J}} d_{kj} r_{j} - \sum_{j \in \mathcal{J}} u_{kj} r_{j} \Big\rvert. 
	\end{align*}
	From constraint \eqref{eq:9c}, the maximum value of $\big\lvert \sum_{j \in \mathcal{J}} d_{kj} r_{j} - \sum_{j \in \mathcal{J}} u_{kj} r_{j} \big\rvert$ equals $1$.
In Big-$M$ linearization, selecting the smallest valid upper bound improves numerical stability and avoids deterioration of branching decisions in mixed-integer solvers~\cite{Dejans2025}. Therefore, we choose
\begin{equation*}
	M = \Big\lceil \max_{k,j} \big\{1,\ \sqrt{p_j^{\text{max}}},\ \sqrt{p_k^{\text{max}} }\Big\} \Big\rceil.
    \IEEEQEDhereeqn
\end{equation*}

\section{Proof of Proposition~\ref{prop:4}}
\label{appendix:4}

Proposition~\ref{prop:2} establishes the equivalence between $f_2$ and $f_3$. Further, from \eqref{eq:19} and \eqref{eq:20}, together with constraints \eqref{eq:25d} and \eqref{eq:25e}, it follows that $f_3$ and $f_4$ are also equivalent. Now, consider the constraints. Constraints \eqref{eq:9c} and \eqref{eq:9e}--\eqref{eq:9g} remain unchanged in both \eqref{eq:17} and \eqref{eq:25}. For constraint \eqref{eq:25b}, constraint \eqref{eq:9b} is equivalently replaced by \eqref{eq:24}, as established in the discussion below \eqref{eq:24}. Further, constraint \eqref{eq:25c} is readily seen to be equivalent to \eqref{eq:9d} under the substitution in \eqref{eq:20}. Finally, as established in the discussion between \eqref{eq:20} and \eqref{eq:23}, the nonlinear constraints in \eqref{eq:20} are equivalent to the linear system in \eqref{eq:23}.
\IEEEQED

\section{Proof of Theorem~\ref{prop:5}}
\label{appendix:5}
	
 First, we prove the monotone increase of the objective $f_2$ in each iteration of the algorithm.  We introduce a subscript $i$ for the iteration index, and the following inequalities hold:
\begin{align*}
	&f_2 (\bm d_{(i)},   \bm r,  \bm u_{(i)}, \bm p^{\mathrm{dl}}_{(i)}, \bm p^{\mathrm{ul}}_{(i)}, \,\bm \chi^{\mathrm{dl}}_{(i)}, \bm \chi^{\mathrm{ul}}_{(i)}, {\bm \xi}^{\mathrm{dl}}_{(i)}, {\bm \xi}^{\mathrm{ul}}_{(i)}) \\
	&\overset{(a)}{=} f_1(\bm d_{(i)},   \bm r,  \bm u_{(i)}, \bm p^{\mathrm{dl}}_{(i)}, \bm p^{\mathrm{ul}}_{(i)}, \,\bm \chi^{\mathrm{dl}}_{(i)}, \bm \chi^{\mathrm{ul}}_{(i)}) \\
	&\overset{(b)}{<} f_1(\bm d_{(i)}, \bm r, \bm u_{(i)}, \bm p^{\mathrm{dl}}_{(i)}, \bm p^{\mathrm{ul}}_{(i)}, \,\bm \chi^{\mathrm{dl}}_{(i+1)}, \bm \chi^{\mathrm{ul}}_{(i+1)}) \\
	&  \overset{(c)}{=}f_2(\bm d_{(i)},  \bm r, \bm u_{(i)}, \bm p^{\mathrm{dl}}_{(i)}, \bm p^{\mathrm{ul}}_{(i)}, \,\bm \chi^{\mathrm{dl}}_{(i+1)}, \bm \chi^{\mathrm{ul}}_{(i+1)}, \bm \xi^{\mathrm{dl}}_{(i+1)}, \bm \xi^{\mathrm{ul}}_{(i+1)}) \\
	&\overset{(d)}{\leq}f_2(\bm d_{(i+1)}, \bm r, \bm u_{(i+1)}, \bm p^{\mathrm{dl}}_{(i+1)}, \bm p^{\mathrm{ul}}_{(i+1)}, \,\bm \chi^{\mathrm{dl}}_{(i+1)}, \bm \chi^{\mathrm{ul}}_{(i+1)}, \\
	&\hspace{190pt} \bm \xi^{\mathrm{dl}}_{(i+1)}, \bm \xi^{\mathrm{ul}}_{(i+1)})
\end{align*}
where $(a)$ is due to $f_2$ being strictly concave in $\xi_{kj}^{\mathrm{dl}}$ and $\xi_{kj}^{\mathrm{ul}}$ (Proposition~\ref{prop:1}).
Therefore, we can obtain unique values as in \eqref{eq:16} that maximizes $f_2$ keeping other variables fixed, and, if we substitute these values into $f_2$, we get back exactly $f_1$.  Thus, the equality follows as we update $\xi_{kj}^{\mathrm{dl}}$ and $\xi_{kj}^{\mathrm{ul}}$ using~\eqref{eq:16} in the algorithm. $(b)$ follows as keeping other variables fixed during each update of $\chi_{kj}^{\mathrm{dl}}$ and $\chi_{kj}^{\mathrm{ul}}$ strictly maximizes the objective. $(c)$ follows by similar logic as in $(a)$. Finally, $(d)$ follows as a joint update of $(\bm d, \bm u, \bm p^{\mathrm{dl}} \bm p^{\mathrm{ul}})$ maximizes the objective $f_2$ keeping other variables fixed.

The sequence of objective values $$\{f_2(\bm d_{(i)},   \bm r,  \bm u_{(i)}, \bm p^{\mathrm{dl}}_{(i)}, \bm p^{\mathrm{ul}}_{(i)}, \,\bm \chi^{\mathrm{dl}}_{(i)}, \bm \chi^{\mathrm{ul}}_{(i)}, {\bm \xi}^{\mathrm{dl}}_{(i)}, {\bm \xi}^{\mathrm{ul}}_{(i)})\}$$ 
generated by the algorithm is strictly increasing and the objective is upper bounded due to power constraints. Therefore, the objective converges to some finite value $f_2^\infty$. The set of iterates generated by the algorithm, denoted by the level set
\begin{align*}
	\mathcal{L} = \{(&\bm d,    \bm u, \bm p^{\mathrm{dl}}, \bm p^{\mathrm{ul}}, \,\bm \chi^{\mathrm{dl}}, \bm \chi^{\mathrm{ul}}, {\bm \xi}^{\mathrm{dl}}, {\bm \xi}^{\mathrm{ul}}) \\
	& \in \mathcal{F} \times \mathbb{R}_+ \times \mathbb{R}_+ \times \mathbb{R} \times \mathbb{R} \\
 &| \, f_2(\bm d,   \bm r,  \bm u, \bm p^{\mathrm{dl}}, \bm p^{\mathrm{ul}}, \,\bm \chi^{\mathrm{dl}}, \bm \chi^{\mathrm{ul}}, {\bm \xi}^{\mathrm{dl}}, {\bm \xi}^{\mathrm{ul}}) \\
 & \geq  f_2(\bm d_{(0)},   \bm r,  \bm u_{(0)}, \bm p^{\mathrm{dl}}_{(0)}, \bm p^{\mathrm{ul}}_{(0)}, \,\bm \chi^{\mathrm{dl}}_{(0)}, \bm \chi^{\mathrm{ul}}_{(0)}, {\bm \xi}^{\mathrm{dl}}_{(0)}, {\bm \xi}^{\mathrm{ul}}_{(0)}) \}
\end{align*}
is compact, so they have at least one accumulation point~\cite[Theorem 6.6.8]{Tao2014} with $\mathcal{F} = \{(\bm d, \bm u, \bm p^{\mathrm{dl}}, \bm p^{\mathrm{ul}}): \textit{satisfying }\eqref{eq:9b}-\eqref{eq:9c} \}$. Let $\{i_n\}$ be an index subsequence that converges such that 
\begin{align*}
	(\bm d_{(i_n)},  &  \bm u_{(i_n)}, \bm p^{\mathrm{dl}}_{(i_n)}, \bm p^{\mathrm{ul}}_{(i_n)}, \,\bm \chi^{\mathrm{dl}}_{(i_n)}, \bm \chi^{\mathrm{ul}}_{(i_n)}, {\bm \xi}^{\mathrm{dl}}_{(i_n)}, {\bm \xi}^{\mathrm{ul}}_{(i_n)}) \\
	& \to (\bm d^{\star}, \bm u^{\star}, \bm p^{\mathrm{dl}{\star}}, \bm p^{\mathrm{ul}{\star}}, \bm \chi^{\mathrm{dl}{\star}}, \bm \chi^{\mathrm{ul}{\star}}, \bm \xi^{\mathrm{dl}{\star}}, \bm \xi^{\mathrm{ul}{\star}}),
\end{align*}
as $n \to \infty$. Since, each update $(\bm \chi^{\mathrm{dl}}_{(i+1)}, \bm \chi^{\mathrm{ul}}_{(i+1)}, \bm \xi^{\mathrm{dl}}_{(i+1)}, \bm \xi^{\mathrm{ul}}_{(i+1)})$ is uniquely determined (cf. \eqref{eq:15} and \eqref{eq:16}), and $f_2$ is continuous with respect to its continuous variables, the continuity of the arg-max implies that
\begin{align*}
	&(\bm \chi^{\mathrm{dl}}_{(i_n+1)}, \bm \chi^{\mathrm{ul}}_{(i_n+1)}, \bm \xi^{\mathrm{dl}}_{(i_n+1)}, \bm \xi^{\mathrm{ul}}_{(i_n+1)})\\
&= \underset{\bm \chi^{\mathrm{dl}} , \bm \chi^{\mathrm{ul}}, \bm \xi^{\mathrm{dl}}, \bm \xi^{\mathrm{ul}}}{\arg\max} f_2(	\bm d_{(i_n)},   \bm r,  \bm u_{(i_n)}, \bm p^{\mathrm{dl}}_{(i_n)}, \bm p^{\mathrm{ul}}_{(i_n)}, \bm \chi^{\mathrm{dl}} , \bm \chi^{\mathrm{ul}}, \bm \xi^{\mathrm{dl}}, \bm \xi^{\mathrm{ul}})  \\
	& \to \underset{\bm \chi^{\mathrm{dl}} , \bm \chi^{\mathrm{ul}}, \bm \xi^{\mathrm{dl}}, \bm \xi^{\mathrm{ul}}}{\arg\max} f_2(\bm d^{\star}, \bm r, \bm u^{\star}, \bm p^{\mathrm{dl}{\star}}, \bm p^{\mathrm{ul}{\star}}, \bm \chi^{\mathrm{dl}} , \bm \chi^{\mathrm{ul}}, \bm \xi^{\mathrm{dl}}, \bm \xi^{\mathrm{ul}} ) \\
	& \qquad= (\bm \chi^{\mathrm{dl}{\star}}, \bm \chi^{\mathrm{ul}{\star}}, \bm \xi^{\mathrm{dl}{\star}}, \bm \xi^{\mathrm{ul}{\star}}),
\end{align*}
as $n \to \infty$.  Additionally, the update $(\bm d_{(i+1)}, \allowbreak  \bm u_{(i+1)},\allowbreak  \bm p^{\mathrm{dl}}_{(i+1)}, \allowbreak \bm p^{\mathrm{ul}}_{(i+1)})$ obtained by solving \eqref{eq:25}, for fixed $(\bm d, \bm u)$, the objective is concave in $\bm p^{\mathrm{dl}}$ and $ \bm p^{\mathrm{ul}}$; therefore, we get
\begin{align*}
	&(\bm d_{(i_n+1)},  \bm u_{(i_n+1)}, \bm p^{\mathrm{dl}}_{(i_n+1)}, \bm p^{\mathrm{ul}}_{(i_n+1)})\\
	& \in \underset{\bm d,    \bm u, \bm p^{\mathrm{dl}}, \bm p^{\mathrm{ul}}}{\arg\max} f_2(	\bm d,   \bm r,  \bm u, \bm p^{\mathrm{dl}}, \bm p^{\mathrm{ul}}, \bm \chi^{\mathrm{dl}}_{(i_n)} , \bm \chi^{\mathrm{ul}}_{(i_n)}, \bm \xi^{\mathrm{dl}}_{(i_n)}, \bm \xi^{\mathrm{ul}}_{(i_n)})  \\
	& \to \underset{\bm d,    \bm u, \bm p^{\mathrm{dl}}, \bm p^{\mathrm{ul}}}{\arg\max} f_2(\bm d, \bm r, \bm u, \bm p^{\mathrm{dl}}, \bm p^{\mathrm{ul}}, \bm \chi^{\mathrm{dl}\star} , \bm \chi^{\mathrm{ul}\star}, \bm \xi^{\mathrm{dl}\star}, \bm \xi^{\mathrm{ul}\star} ) \\
	& \qquad\in (\bm d^\star,    \bm u^\star, \bm p^{\mathrm{dl}\star}, \bm p^{\mathrm{ul}\star}).
\end{align*}
This implies for fixed $(\bm d, \bm u)$, the point $(\bm d^\star,    \bm u^\star, \bm p^{\mathrm{dl}\star}, \bm p^{\mathrm{ul}\star})$ is stationary in $\bm p^{\mathrm{dl}}$ and $ \bm p^{\mathrm{ul}}$. Finally, putting it all together, the accumulation point $(\bm d^{\star}, \bm u^{\star}, \bm p^{\mathrm{dl}{\star}}, \bm p^{\mathrm{ul}{\star}}, \bm \chi^{\mathrm{dl}{\star}}, \bm \chi^{\mathrm{ul}{\star}}, \bm \xi^{\mathrm{dl}{\star}}, \bm \xi^{\mathrm{ul}{\star}})$ is stationary in $\bm \chi^{\mathrm{dl}}, \bm \chi^{\mathrm{ul}}, \bm \xi^{\mathrm{dl}}, \bm \xi^{\mathrm{ul}}, \bm p^{\mathrm{dl}}$, and $ \bm p^{\mathrm{ul}}$,  for fixed $(\bm d, \bm u)$. Further, as we search over all possible $\bm r$, $\bm r^\star$ is optimal. 	 \IEEEQED

\balance
\bibliographystyle{IEEEtran}
\bibliography{IEEEtrancfg,IEEEabrv,ref}

\end{document}